\documentclass[11pt,a4paper]{article}\def\today{February 12, 2009}
\usepackage{amssymb}
\usepackage{amsmath}
\usepackage{theorem}
\usepackage{latexsym}
\usepackage{graphicx}
\reversemarginpar
\theoremheaderfont{\bf}
\theorembodyfont{\sl}
\newtheorem{theo}{Theorem}[section]
{\theorembodyfont{\rm} \newtheorem{defi}[theo]{Definition}}
{\theorembodyfont{\rm} \newtheorem{exa}[theo]{Example}}
{\theorembodyfont{\rm} \newtheorem{rem}[theo]{Remark}}
\newtheorem{prop}[theo]{Proposition}
\newtheorem{cor}[theo]{Corollary}

{\theorembodyfont{\rm} }

{\theorembodyfont{\rm}}
{\theorembodyfont{\rm}}
\newenvironment{proof}{{\sc Proof:}}{\mbox{}\hfill$\Box$\par}
\newcommand{\eqnref}[1]{~\mbox{$(${\rm \ref{#1}}$)$}}

\renewcommand{\theequation}{\thesection.\arabic{equation}}

\newcommand{\junk}[1]{}

\newcommand{\N}{{\mathbb N}}
\newcommand{\F}{{\mathbb F}}
\newcommand{\Z}{{\mathbb Z}}
\newcommand{\C}{{\mathbb C}}
\newcommand{\cC}{{\mathcal C}}
\newcommand{\cM}{{\mathcal M}}
\newcommand{\cR}{{\mathcal R}}
\newcommand{\cS}{{\mathcal S}}
\newcommand{\cCbar}{\mbox{$\overline{\cC}$}}
\newcommand{\cCpol}{\mbox{${\mathcal C}_{\text{pol}}$}}
\newcommand{\cCbarpol}{\mbox{$\overline{{\mathcal C}}_{\text{pol}}$}}
\newcommand{\Gbar}{\mbox{$\bar{G}$}}

\newcommand{\CC}{convolutional code}

\newcommand{\rank}{\text{rk}\,}

\newcommand{\cMzn}{\mbox{$\cM_{n,z}$}}
\newcommand{\zME}{\mbox{$z$}\text{ME}}
\newcommand{\imR}{\mbox{$\text{im}_{\mathcal R}$}}
\newcommand{\imFz}{\mbox{$\text{im}_{{\mathbb F}[z]}$}}
\newcommand{\edge}[2]{\mbox{$-\!\!\!-\!\!\!\longrightarrow$}\hspace{-2em}
     \raisebox{.3ex}{$\begin{array}{c}{\scriptscriptstyle{#1}}\\[-1.4ex]{\scriptscriptstyle{#2}}\end{array}$}\hspace{.6em}}

\newcommand{\im}{\mbox{\rm im}\,}

\newcommand{\dist}{\mbox{\rm dist}}
\newcommand{\del}{\mbox{\rm del}\,}
\newcommand{\wt}{\mbox{${\rm wt}$}}
\newcommand{\we}{\mbox{\rm we}}
\newcommand{\WAM}{\mbox{\rm WAM}}

\newcommand{\T}{\mbox{$\!^{\sf T}$}}

\newcounter{alp}
\newcounter{ara}
\newcounter{rom}
\newenvironment{romanlist}{\begin{list}{(\roman{rom})\hfill}{\usecounter{rom}
     \topsep0ex \labelwidth.7cm \leftmargin.7cm \labelsep0cm
     \rightmargin0cm \parsep0ex \itemsep.6ex
     \partopsep1.6ex}}{\end{list}}
\newenvironment{alphalist}{\begin{list}{(\alph{alp})\hfill}{\usecounter{alp}
     \topsep0ex \labelwidth.7cm \leftmargin.7cm \labelsep0cm
     \rightmargin0cm \parsep0ex \itemsep.6ex
     \partopsep1.6ex}}{\end{list}}
\newenvironment{arabiclist}{\begin{list}{(\arabic{ara})\hfill}{\usecounter{ara}
     \topsep0ex \labelwidth.7cm \leftmargin.7cm \labelsep0cm
     \rightmargin0cm \parsep0ex \itemsep.6ex
     \partopsep1.6ex}}{\end{list}}

\title{On Isometries for Convolutional Codes}
\date\today
\author{Heide Gluesing-Luerssen\footnote{University of Kentucky, Department of Mathematics,
       715 Patterson Office Tower, Lexington, KY 40506-0027, USA; heidegl@ms.uky.edu}
       }

\begin{document}
\maketitle

{\bf Abstract:} In this paper we will discuss isometries and strong isometries for convolutional codes.
Isometries are weight-preserving module isomorphisms whereas strong
isometries are, in addition, degree-preserving. Special cases of these maps are certain types of monomial transformations.
We will show a form of MacWilliams Equivalence Theorem, that is, each isometry between convolutional codes is given by
a monomial transformation.
Examples show that strong isometries cannot be characterized this way, but special attention paid to the weight adjacency matrices
allows for further descriptions.
Various distance parameters appearing in the literature on convolutional codes will be discussed as well.



\section{Introduction and Basic Notions}\label{SS-Intro}
\setcounter{equation}{0}

One of the most famous results in the theory of linear block codes is MacWilliams' Equivalence Theorem~\cite{MacW62,MacW63}; see also
\cite[Sec.~7.9]{HP03}.
It tells us that two block codes are isometric if and only if they are monomially equivalent.
Stated more precisely, weight-preserving isomorphisms between codes are given by a permutation and rescaling of the
coordinates.
Hence the intrinsic notion of isometry coincides with the extrinsic notion of monomial equivalence and this settles the
question of a classification of block codes over fields with respect to their error-corrrecting properties.
Since the discovery of the importance of linear block codes over~$\Z_4$ for nonlinear binary codes, the result has enjoyed various
generalizations to block codes over certain finite rings, see for instance the articles \cite{WaWo96,Wo99a,GrSch00,DiLP04a}.

For convolutional code  such a result is not yet known.
In other words, a classification taking all relevant parameters of the codes into account has not yet been established and it is not yet clear as to when two such codes may be declared the same with respect to their error-correcting performance.
In this paper, we will make a step in this direction by studying isometries, that is, weight-preserving $\F[z]$-module isomorphisms between convolutional codes in $\F[z]^n$.
It is immediate that the notion of isometry is too weak in order to give a meaningful classification of convolutional codes.
Indeed, one can readily present convolutional codes of positive degree that are isometric to certain block codes.
For this reason we will also consider degree-preserving isometries, in other words, isometries that also preserve
the length of the codeword sequence.
They will be called strong isometries throughout this paper.

In the next section we will recall various distance parameters known for convolutional codes, see, e.~g., \cite{JoZi99, HJZZ99,JPB90}, and in
Section~\ref{SS-Iso} we will discuss which of them are preserved under (strong) isometries.
In Section~\ref{SS-WAM} we will also introduce a crucial invariant of convolutional codes, the weight adjacency matrix (WAM, for short).
This matrix has been studied in detail in the literature
\cite{McE98a,GL05p,GS07}.
It is indexed by the states of a chosen state transition diagram and contains, for each pair of states, the weight enumerator
polynomial of the set of outputs corresponding to the respective state transition.
Factoring out the group of state space isomorphisms the WAM turns into an invariant that
contains an abundance of information about the code.
Indeed, all the distance parameters mentioned above can be deduced from the WAM.
However, the WAM is an even more detailed invariant of the code.
Indeed, in Section~\ref{SS-Iso} we will present an example showing that (strongly isometric) codes sharing all those distance parameters do not
necessarily have the same WAM.
Further proof of the strength of the WAM is provided by the MacWilliams Duality Theorem which tells us that the WAM of the dual code can be computed
directly from the WAM of the primary code without having to compute weights of branches in the state transition diagram of the dual
code~\cite{GS08,GS08p}.

In Section~\ref{SS-MacWE} we will show that, similar to the MacWilliams Equivalence Theorem for block codes, two convolutional codes are
isometric if and only if they differ only by a $z$-monomial transformation, that is, by a permutation of the codeword coordinates and a rescaling with
monomials of the form $\alpha z^s,\,\alpha\in\F^*,\,s\in\Z$.
However, even though this result is very useful due to its explicit description of all isometries, it is not quite the answer one is looking for
because isometries are too weak of a notion in order to classify convolutional codes in a meaningful way.
But the characterization just mentioned is the best one can hope for because, as we will see in
Section~\ref{SS-Iso}, no such result can be expected for strong isometries.
Indeed, there exist strongly isometric codes that are not necessarily monomially equivalent in the classical sense (allowing, besides
permutations, only rescaling by nonzero constants) --- even if the codes share the same WAM.
Fortunately, there are two positive results concerning strongly isometric codes.
Firstly, codes with only positive Forney indices and sharing the same WAM are even monomially equivalent and no a priori knowledge about a strong
isometry is needed; this has been proven earlier in~\cite{GS07}.
Secondly, strongly isometric {\sl encoders\/} sharing the same WAM with respect to the same fixed basis of their common state space are
monomially equivalent.
Stated in other words, two strongly isometric codes for which there exists a pair of strongly isometric bases sharing the same state
transitions and the same WAM are monomially equivalent.

All the above shows that a classification of convolutional codes is not as straightforward as for block codes.
Due to the abundance of information about the performance of a code stored in its WAM, we think it is worthwhile to investigate those
strong isometries that also preserve the WAM of a code.
An explicit characterization of those maps would enable us to tell whether two codes may be regarded completely the same with respect
to their error-correcting capabilities.
We have to leave this as an open question to future research.

Finally, in the last section of this paper we will discuss generalizations of our results in two different directions.
Firstly, we will briefly consider weight-preserving $\F$-linear (but not necessarily $\F[z]$-linear) isomorphisms between convolutional
codes.
This much weaker notion of isometry raises plenty of interesting questions, which, again, we have to leave open to future research.
First answers can be found in the monograph~\cite{Pi88b}.
Secondly, we will discuss (strong) isometries for convolutional codes defined in the setting of rational functions and Laurent series.
We will see that the results from the polynomial setting extend in an expected way to this larger context.

Let us now close the introduction with recalling the basic notions of convolutional coding theory as needed throughout the paper.
From now on let
\begin{equation}\label{e-Fq}
    \F=\F_q\text{ be a finite field with~$q$ elements}.
\end{equation}
We will need the {\sl weight enumerator\/} of sets $S\subseteq \F^n$.
It is given as
\begin{equation}\label{e-we}
  \we(S):=\sum_{i=0}^n \lambda_iW^i\in\C[W], \text{ where }
  \lambda_i:=\#\{v\in S\mid \wt(v)=i\},
\end{equation}
where, of course, $\wt(v)$ denotes the Hamming weight of $v\in\F^n$.
The weight enumerator $\we(\cC)$ of a block code $\cC\subseteq\F^n$ has been investigated intensively in the
block coding literature.
For instance, the famous MacWilliams Identity Theorem~\cite{MacW63} tells us how to completely derive
$\we(\cC^{\perp})$ from $\we(\cC)$, where $\cC^{\perp}$ is the dual of~$\cC$ with
respect to the standard inner product on~$\F^n$.

Throughout the main part of this article a convolutional code will be defined in the polynomial setting.
Only in Section~\ref{SS-Gen} we will briefly address a broader context.
Thus, a {\sl convolutional code of length\/} $n$ is a submodule~$\cC$ of $\F[z]^n$ of the form
\[
    \cC=\im G:=\{uG\,\big|\, u\in\F[z]^k\}
\]
where~$G$ is a {\sl basic\/} matrix in $\F[z]^{k\times n}$, i.~e.
\begin{equation}\label{e-Gbasic}
  \rank G(\lambda)=k \text{ for all }\lambda\in\overline{\F},
\end{equation}
with $\overline{\F}$ being an algebraic closure of~$\F$.
We call such a matrix $G$ an {\sl encoder}, and the number
\begin{equation}\label{e-degG}
   \deg(\cC):=\deg(G):=\max\{\deg(M)\mid M\text{ is a $k$-minor of }G\}
\end{equation}
is said to be the {\sl degree\/} of the encoder~$G$ or of the code~$\cC$.
It is clear that for two basic matrices $G,\,G'\in\F[z]^{k\times n}$ one has
$\im G=\im G'$ if and only if $G'=UG$ for some $U\in GL_k(\F[z])$.
For each basic matrix the sum of its row degrees is at least $\deg(G)$, where the degree of a
polynomial row vector is defined as the maximal degree of its entries.
A matrix $G\in\F[z]^{k\times n}$ is said to be {\sl reduced\/} if the sum of its row degrees equals
$\deg(G)$.
Among the many characterizations of reducedness that can be found in, e.~g., \cite[Main~Thm.]{Fo75} or
\cite[Thm.~A.2]{McE98}, one will be particularly useful for us.
If~$G\in\F[z]^{k\times n}$ is reduced with row degrees $\nu_1,\ldots,\nu_k$ then
\begin{equation}\label{e-reduced}
   \deg(uG)=\max\{\deg u_i+\nu_i\mid i=1,\ldots,k\} \text{ for all }u=(u_1,\ldots,u_k)\in\F[z]^k.
\end{equation}
It is well known~\cite[p.~495]{Fo75} that each convolutional code admits a reduced encoder.
The row degrees of a reduced encoder are, up to ordering, uniquely determined by the code and are called the
{\sl Forney indices\/} of the code or of the encoder.
The maximal Forney index is called the {\sl memory\/} of the code.
It follows that a convolutional code has a constant encoder matrix if and only if the degree is zero.
In that case the code is, in a natural way, a block code.
Finally, the dual of a code $\cC\subseteq\F[z]^n$ is defined as usual as
\begin{equation}\label{e-Cdual}
  \cC^{\perp}=\{w\in\F[z]^n\mid wv\T=0\text{ for all }v\in\cC\}.
\end{equation}
It is well known that if~$\cC$ is a $k$-dimensional code of degree~$\delta$, then $\cC^{\perp}$ is a code of dimension $n-k$ and
degree~$\delta$.

Besides these algebraic notions the main concept in error-control coding is the weight.
For a vector $v=\sum_{t=0}^N v^{(t)}z^t\in\F[z]^n$, $v^{(t)}\in\F^n$, we define
its weight as $\wt(v)=\sum_{t=0}^N \wt(v^{(t)})$.
The {\sl distance\/} of a (block or convolutional) code~$\cC$ is
$\dist(\cC)=\min\{\wt(v)\,|\, v\in\cC,\,v\not=0\}$.
In convolutional coding theory several other distance parameters have been introduced.
They all give detailed information about the performance of the code under various decoding algorithms.
We will present some of those parameters below in Definition~\ref{D-distparam}.

\section{The Weight Adjacency Matrix and Distance Parameters}\label{SS-WAM}
\setcounter{equation}{0}
Let $\cC=\im G$ be a \CC\ of degree~$\delta$ and with reduced encoder $G\in\F[z]^{k\times n}$.
It is well-known that the encoding process can be described in terms of a minimal
state space system, that is, there exist matrices
$(A,B,C,D)\in\F^{\delta\times\delta+k\times\delta+\delta\times n+k\times n}$ such that
\begin{equation}\label{e-iso}
  v=uG
        \Longleftrightarrow
        \left\{\begin{array}{rcl} x_{t+1}&=&x_tA+u_tB\\v_t&=&x_tC+u_tD\end{array}
          \;\text{ for all }t\geq0\right\} \text{ where }x_0=0
\end{equation}
for any $u=\sum_{t\geq0}u_tz^t\in\F[z]^k$ and $v=\sum_{t\geq0}v_tz^t\in\F[z]^n$,
see also \cite[Thm.~2.3]{GL05p} and \cite{McE98a}.
A particular realization will be introduced in Proposition~\ref{P-CCF} below.
We call the space $\F^{\delta}$, where the states $x_t$ assume their values, the
{\em state space\/} of the code, and the matrix quadruple $(A,B,C,D)$ is called a
{\em minimal state space realization}.
In the sequel we will review a few basic properties of minimal state space realizations,
for details  see, e.~g., \cite[Sec.~2]{GL05p}.
In general a code has many state space realizations.
Each realization uniquely determines the encoder~$G$ satisfying\eqnref{e-iso} via
$G=B(z^{-1}I-A)^{-1}C+D$.
The dynamical system given by the state space realization\eqnref{e-iso} can be visualized by a state
transition diagram, which is defined as the directed and labeled graph with the vertices
given by the states in $\F^{\delta}$ and the labeled edges given by the set
\[
   \Big\{X\edge{u}{v}\;Y\,\Big|\, X,\,Y\in\F^{\delta}, u\in\F^k,\, v\in\F^n: Y=XA+uB,\, v=XC+uD\Big\}.
\]
Using\eqnref{e-iso} we may identify the codewords $v=\sum_{t=0}^Nv_tz^t=(\sum_{t=0}^Nu_tz^t)G,\,u_t\in\F^k$, of degree~$N$
with the path of length~$N+1$ through the state transition diagram given by
\begin{equation}\label{e-path}
   x_0=0\edge{u_0}{v_0}x_1\edge{u_1}{v_1}x_2\ldots x_N\edge{u_N}{v_N}x_{N+1}=0.
\end{equation}
Here the state $x=\sum_{t=0}^N x_tz^t\in\F[z]^{\delta}$ can either be computed from\eqnref{e-iso} or via
\begin{equation}\label{e-state}
     x=uB(z^{-1}I-A)^{-1}=uB\sum_{t=1}^{\delta}A^{t-1}z^t\in\F[z]^{\delta}.
\end{equation}
We call the codeword~$v\in\cC$ of degree~$N$ {\sl atomic\/} if $x_t\not=0$ for all $t=1,\ldots,N$ in the associated
path\eqnref{e-path}.
Alternatively,~$v$ is atomic if and only if its constant coefficient is nonzero and~$v$ cannot be
written as the sum of two nonzero codewords $\hat{v},\,\tilde{v}\in\cC$ such that $\deg\hat{v}\leq L$ and
$\tilde{v}\in z^{L+1}\F[z]^n$ for some $L\in\N_0$.
Thus, being atomic is a codeword property that does not depend on the choice of the realization.

For the purpose of this paper it will be sufficient to consider one particular form of
state space realization, the controller canonical form.
This realization is standard in systems and coding theory.
Since the precise form will play an important role later on, we think it is worthwhile to
present it explicitly; see also \cite[Prop.~2.3]{GS07}.

\begin{prop}\label{P-CCF}
Let $G\in\F[z]^{k\times n}$ be a basic and reduced matrix with row degrees $\nu_1,\,\ldots,\nu_k$.
Put $\delta:=\sum_{i=1}^k\nu_i$.
Let~$G$ have rows $g_i=\sum_{\ell=0}^{\nu_i}g_{i,\ell}z^{\ell},\,i=1,\ldots,k,$ where
$g_{i,\ell}\in\F^n$.
For $i=1,\ldots,k$ define the matrices
\[
 A_i=\left(\begin{smallmatrix} 0&1& & \\ & &\ddots& \\& & &1\\ & & &0\end{smallmatrix}\right)
      \in\F^{\nu_i\times\nu_i},\
 B_i=\begin{pmatrix}1&0&\cdots&0\end{pmatrix}\in\F^{\nu_i},\
 C_i=\begin{pmatrix}g_{i,1}\\ \vdots\\ g_{i,\nu_i}\end{pmatrix}\in\F^{\nu_i\times n}.
\]
Then the {\sl controller canonical form (CCF)\/} of~$G$ is defined as
the matrix quadruple
$(A,B,C,D)\in\F^{\delta\times\delta}\times\F^{k\times\delta}\times
             \F^{\delta\times n}\times\F^{k\times n}$
where
\[
   A=\left(\begin{smallmatrix} A_1&  & \\ &\ddots &\\ & &A_k\end{smallmatrix}\right),\:
   B=\left(\begin{smallmatrix}
            B_1& &\\ &\ddots & \\ & &B_k\end{smallmatrix}\right),\:
   C=\left(\begin{smallmatrix}C_1\\ \vdots\\[.5ex]C_k\end{smallmatrix}\right),\:
   D=\left(\begin{smallmatrix}g_{1,0}\\[.4ex] \vdots\\[.6ex]g_{k,0}\end{smallmatrix}\right)=G(0).
\]
In the case where $\nu_i=0$ the $i$th block is missing and in~$B$ a zero row occurs.
The CCF forms a minimal state space realization of the code $\cC=\im G$, that is,\eqnref{e-iso}
is satisfied.
\end{prop}

Now we are in a position to summarize various distance parameters for convolutional codes.
The following parameters are, in different ways, closely related to the error-correcting performance of~$\cC$ and
are studied intensively in the engineering-oriented literature, see, e.~g.,~\cite{JPB90,HJZZ99} and
\cite[Sec.~3.2]{JoZi99}.
For instance, the weight enumerator, defined below, may be used to derive an upper bound for the burst error
probability of the code used on a binary symmetric channel with maximum-likelihood decoding,
see \cite{Vi71} and \cite[Thm.~4.2]{JoZi99}.

For a vector polynomial $v=\sum_{t=0}^Nv_tz^t\in\F[z]^n$ let $v_{[0,j]}=\sum_{t=0}^j v_tz^t$ be the
{\sl truncation\/} of~$v$ at time~$j$.

\begin{defi}\label{D-distparam}
Let $\cC=\im G\subseteq\F[z]^n$ be a convolutional code of degree~$\delta$ with reduced encoder~$G$ and minimal state
space realization $(A,B,C,D)$, which need not be the CCF of~$G$.
Let $\nu_1,\ldots,\nu_k$ be the Forney indices of~$\cC$ and let $m:=\max\{\nu_1,\ldots,\nu_k\}$
be its memory. For $j\in\N_0$ define
\[
    \cS_j:=\{v\in\cC\mid (x_i,x_{i+1})\not=(0,0)\text{ for all }i=0,\ldots,j\},
\]
where, as usual, $x\in\F[z]^{\delta}$ denotes the state sequence associated with the codeword~$v\in\cC$.
\begin{alphalist}
\item For $j\in\N_0$ define
      \begin{align*}
             &\text{{\sl $j$th column distance\/}:}& &d^c_j:=\min\{\wt(v_{[0,j]})\mid v\in\cC,\,v_0\not=0\},\\[.6ex]
             &\text{{\sl $j$th extended row distance\/}:}& &\hat{d}^r_j:=\min\{\wt(v)\mid v\text{ atomic}, \deg(v)=j\},\\[.6ex]
             &\text{{\sl $j$th active column distance\/}:}& &a^c_j:=\min\{\wt(v_{[0,j]})\mid v\in\cS_j\},\\[.6ex]
             &\text{{\sl $j$th active row distance\/}:}& &a^r_j:=\min\{\wt(v_{[0,m+j]})\mid v=uG\in\cS_{j},\,\deg(u)=j\},\\[.6ex]
             &\text{{\sl $j$th active segment distance\/}:}& &a^s_j:=\min\Big\{\wt(v_{[m,m+j]})\,\Big|\,
                             \begin{array}{l} (x_i,x_{i+1})\not=(0,0)\\ \text{for }i=m,\ldots, m+j\end{array}\Big\}
      \end{align*}
      and for $j\geq\min\{1,\nu_1,\ldots,\nu_k\}$ define
      \begin{align*}
             &\text{{\sl $j$th active burst distance\/}:}\hspace*{1em}& &a^b_j:=\min\{\wt(v_{[0,j]})\mid v\in\cS_j,\,x_{j+1}=0\}.\hspace*{5.8em}
      \end{align*}
\item The {\sl weight enumerator\/} of~$\cC$ is defined as
      \[
         \we(\cC)=1+\sum_{l=1}^{\infty}\Omega_lL^l\in\C[\![W,L]\!],
         \text{ where }\Omega_l=\sum_{\alpha=1}^{\infty}\omega_{l,\alpha}W^{\alpha}=\sum_{\alpha=1}^{nl}\omega_{l,\alpha}W^{\alpha}\in\C[W]
      \]
      with
      $\omega_{l,\alpha}:=\big|\{v\in\cC\mid v\text{ atomic}, \deg v=l-1,\, \wt(v)=\alpha\}\big|$
      for all $\alpha,\,l\in\N$.
      Thus, $\Omega_l$ is the weight enumerator of the set of atomic codewords of degree~$l$.
\end{alphalist}
\end{defi}

A few comments are in order.
First of all observe that the sets~$\cS_j$ consists of those codewords for which the according path through the state transition
diagram does not contain edges of the form $x_i=0\,\edge{u_i}{v_i}\,0=x_{i+1}$ for any $0\leq i\leq j$.
In particular, the path diverges immediately from the zero state, that is, $u_0\not=0$.
Secondly, notice that for the active row distance $\deg(u)=j$ implies $\deg(uG)\leq m+j$ and thus $v=v_{[0,m+j]}$ is a codeword.
Likewise for the active burst distance the requirement $x_{j+1}=0$ implies that $v_{[0,j]}$ is a codeword,
see\eqnref{e-path}.
This in turn implies that~$j$ is not smaller than the smallest Forney index of the code for such a codeword to exist,
see\eqnref{e-reduced}.
Since, by definition, $\{v\in\cS_0\mid x_1=0\}=\emptyset$ all this leads to the restriction on~$j$ as given in~(a) above
for the active burst distances.

It should be noted that the column distances satisfy $d^c_0\leq d^c_1\leq\ldots\leq\dist(\cC)$ and $d^c_j=\dist(\cC)$
for sufficiently large~$j$.
Likewise, the extended row distances satisfy $\hat{d}^r_j=\dist(\cC)$ for sufficiently large~$j$.
Thus, the sequences $(d^c_j)_{j\in\N_0}$ and $(\hat{d}^r_j)_{j\in\N_0}$ are stationary.
This is not the case for the active distances: due to the restriction on the state sequence enforced by the sets~$\cS_j$ these
distance parameters ``stay active'' for all $j\in\N_0$.
In the following remark we discuss which of the parameters introduced above are invariants of the code.

\begin{rem}\label{R-distparam}
\begin{arabiclist}
\item The sets $\cS_j$ do not depend on the chosen realization.
      This can be seen as follows.
      From \cite[Thm.~2.6]{GS07} it is known that if $(A,B,C,D)$ and $(\bar{A},\bar{B},\bar{C},\bar{D})$ are two minimal realizations
      of~$\cC$, then there exist matrices $T\in GL_{\delta}(\F),\, U\in GL_k(\F)$, and $M\in\F^{\delta\times k}$ such that
      $(\bar{A},\bar{B},\bar{C},\bar{D})=\big(T^{-1}(A-MB)T,UBT,T^{-1}(C-MD),UD)$.
      Now it is straightforward to check that for any $(x_t,\,x_{t+1},\,u_t,\,v_t)\in\F^{\delta+\delta+k+n}$ one has
      \[
        \big[x_{t+1}=x_tA+u_tB,\ v_t=x_tC+u_tD\big] \Longleftrightarrow
        \big[x_{t+1}T=x_tT\bar{A}+\bar{u}_t\bar{B},\ v_t=x_tT\bar{C}+\bar{u}_t\bar{D}\big],
      \]
      where $\bar{u}_t=u_tU^{-1}+x_tMU^{-1}$.
      This shows that for a given codeword $v\in\cC$ the polynomial
      $x=\sum_{t\geq0}x_tz^t\in\F[z]^{\delta}$ is the state sequence with respect to $(A,B,C,D)$ if and only if
      $xT=\sum_{t\geq0}x_tTz^t$ is the state sequence with respect to $(\bar{A},\bar{B},\bar{C},\bar{D})$.
      Since $x_t\not=0\Longleftrightarrow x_tT\not=0$ this proves that $\cS_j$ does not depend on the choice of the realization.
\item All distance parameters defined in Definition~\ref{D-distparam}, with the exception of the active row distances, are invariants
      of the code and do not depend on the choice of the realization.
      This is obvious for the column distances, the extended row distances as well as the weight enumerator and follows from~(1)
      for the active distances.
\item The active row distances are invariants of the reduced encoder and do not depend on the choice of the realization.
      This follows directly from the definition and~(1).
      But, due to the degree constraint on the message polynomial $u$, they are not invariants of the code.
      Indeed, this is due to the trivial fact that for a given codeword the degree of the associated message
      depends on the choice of the reduced encoder.
      For example, consider $\cC=\im G=\im G'$, where
      \[
         G=\begin{pmatrix}1&z&0\\1&1&1\end{pmatrix},\
         G'=\begin{pmatrix}1&z\\0&1\end{pmatrix}G=\begin{pmatrix}1+z&0&z\\1&1&1\end{pmatrix}\in\F_2[z]^{2\times3}.
      \]
      Using the CCF's of~$G$ and~$G'$ it is easy to see that the first active row distances are given by
      $a^r_1(G)=\wt\big((1,z)G\big)=3$ and $a^r_1(G')=\wt\big((1,z)G'\big)=\wt\big((1,0)G\big)=2$.
\end{arabiclist}
\end{rem}

As we will show next, the extended row distances and the active burst distances can be retrieved, at least in theory, from the
weight enumerator.

\begin{prop}\label{P-extrowactiveburst}
For a nonzero polynomial $a=\sum_{i=0}^Ra_i W^i\in\C[W]$ define the delay to be
$\del(a):=\min\{i\mid a_i\not=0\}$ and put $\del(0)=\infty$.
Then the extended row distances and the active burst distances are given by
\[
  \hat{d}^r_j=\del(\Omega_{j+1})\ \text{ and }\
  a^b_j=\min\Big\{\sum_{l=1}^r\del(\Omega_{M_l})\,\Big|\, r\in\N,\,M_l\in\N_{\geq2},\,\sum_{l=1}^r M_l=j+1\Big\}.
\]
\end{prop}
\begin{proof}
The statement for the extended row distances is clear from Definition~\ref{D-distparam}.
For the active burst distances we show first that $\{v_{[0,j]}\mid v\in\cS_j,\,x_{j+1}=0\}=\{v\in\cS_j\mid \deg v=j\}$.
Indeed, the inclusion ``$\supseteq$'' is obvious.
As for ``$\subseteq$'' notice first that $v_{[0,j]}\in\cC$ due to $x_{j+1}=0$, see\eqnref{e-path}.
Moreover,\eqnref{e-iso} implies $(0,\,v_j)=(x_j,u_j)\big(\begin{smallmatrix}A&C\\B&D\end{smallmatrix}\big)$.
Since the rightmost matrix has full row rank \cite[Thm.~2.4]{GS07} and $x_j\not=0$ due to the definition of
the set~$\cS_j$, we obtain $v_j\not=0$ and thus $v_{[0,j]}$ is in the set on the right hand side.
This proves ``$\subseteq$''.
In terms of paths the set~$\{v\in\cS_j\mid \deg v=j\}$ consists exactly of all paths as in\eqnref{e-path} of
length~$j+1$ that do not pass through the zero state at consecutive time instants and for which $v_j\not=0$.
Consider such a path and suppose it passes through the zero state exactly at the time instants $t_0=0,\,t_1,\ldots,\,t_r=j+1$.
Then each subpath between $x_{t_{i-1}}=0$ and $x_{t_i}=0$ represents an atomic codeword~$v^{(i)}\in\cC$ of degree
$t_i-t_{i-1}-1\geq1$.
Moreover, $\sum_{i=1}^r(\deg v^{(i)}+1)=j+1$.
Now minimizing the weight of all such codewords is equivalent to minimizing the weight of each atomic subpath
and the result follows from the definition of~$\Omega_l$ in Definition~\ref{D-distparam}(b).
\end{proof}

Examples later on will show that the other distance parameters can, in general, not be retrieved from the weight enumerator.
However, all the distance parameters introduced above, except for the active row distances, can be computed from one
parameter, the weight adjacency matrix.
This matrix is defined via the state space system\eqnref{e-iso}, where $(A,B,C,D)$ is the CCF
of a reduced encoder.
For each pair of states $(X,Y)\in\F^{\delta}\times\F^{\delta}$ it counts all the weights of the outputs $v=XC+uD$ corresponding to
inputs~$u$ that realize a state transition from~$x_t=X$ to~$x_{t+1}=Y$, that is, satisfy the identity $Y=XA+uB$.
This leads to the following definition, see also \cite[Sec.~2]{McE98a} and \cite[Def.~3.4]{GL05p}.
Note that, due to\eqnref{e-Fq}, the system may assume $q^{\delta}$ different states~$X$.
Recall the weight enumerator defined in\eqnref{e-we}.

\begin{defi}\label{D-Lambda}
Let $G\in\F[z]^{k\times n}$ be a reduced encoder such that $\deg(G)=\delta$ and let
$(A,B,C,D)$ be the CCF of~$G$.
The {\sl weight adjacency matrix (WAM)\/} of~$G$ is defined to be the matrix
$\Lambda(G)\in\C[W]^{q^{\delta}\times q^{\delta}}$ that is
indexed by the state pairs $(X,Y)\in\F^{\delta}\times\F^{\delta}$ and has the entries
\begin{equation}\label{e-LambdaXY}
     \Lambda(G)_{X,Y}:=\we\{XC+uD\mid u\in\F^k: Y=XA+uB\}\in\C[W]\text{ for }(X,Y)\in\F^{\delta}\times\F^{\delta}.
\end{equation}
\end{defi}

If $\delta=0$ the matrices $A,\,B,\,C$ do not exist and $D=G$.
As a consequence, $\Lambda(G)=\Lambda(G)_{0,0}=\we(\cC)$ is the weight enumerator
of the block code $\{uG\mid u\in\F^k\}$.

By definition, the WAM of a given encoder is the adjacency matrix of the weighted state transition diagram
as considered in \cite[Sec.~3.10]{JoZi99} and \cite[Sec.~2]{McE98a}.
Its properties have been studied in detail in the papers~\cite{GL05p,GS08,GS08p}.
In a slightly different form the weight adjacency matrix appears also in other papers on convolutional
coding theory, see, e.~g., \cite[Sec.~3.10]{JoZi99}.
In the papers~\cite{GS08,GS08p} a MacWilliams Identity for convolutional codes and their
duals is established.
The latter makes sense only because the weight adjacency matrix can  be turned into an
invariant of the code.
Indeed, we have the following result, see~\cite[Sec.~4]{GL05p}.

\begin{theo}\label{T-Lambdaunique}
Let $\cC\subseteq\F[z]^n$ be a code of degree~$\delta$ and let~$G$ and~$\Gbar$ be any two reduced
encoders of~$\cC$.
Then there exists a state space isomorphism $T\in GL_{\delta}(\F)$ such that
\begin{equation}\label{e-Lambdatrafo}
  \Lambda(G)_{X,Y}=\Lambda(\Gbar)_{XT,YT}\text{ for all } (X,Y)\in\F^{\delta}\times\F^{\delta}.
\end{equation}
The orbit
\[
  \WAM(\cC):=
   \{\Lambda'\mid \exists \,T\in GL_{\delta}(\F):\,\Lambda'_{X,Y}=\Lambda(G)_{XT,YT}
   \text{ for all }(X,Y)\in\F^{\delta}\times\F^{\delta}\}
\]
is called the {\sl weight adjacency matrix\/} (WAM) of the code~$\cC$.
\end{theo}

It is worth being stressed that one can, of course, define the WAM with respect to any minimal realization rather
than the CCF.
In that case one obtains the same result as in Theorem~\ref{T-Lambdaunique} for the WAM's of a code, which in turn
leads to the same invariant as in that theorem; for details consult~\cite{GL05p}.

In the sequel we will show that $\WAM(\cC)$ is the most detailed weight enumerating
object for a given code in the sense that all distance parameters introduced earlier
can be derived from it.
Again, this does not include the active row distances because they depend on the chosen encoder whereas
$\WAM(\cC)$ does not.

Let $\cC=\im G\subseteq\F[z]^n$ be a code of degree~$\delta$ with reduced encoder~$G$.
Furthermore, let $(A,B,C,D)$ be the CCF of~$G$ and let $\Lambda=\Lambda(G)$.
Define the following {\sl reduced WAM's}
\begin{align}
   &\tilde{\Lambda}\in\C[W]^{q^{\delta}\times q^{\delta}}\text{ where }\tilde{\Lambda}_{0,0}:=0\text{ and }
     \tilde{\Lambda}_{X,Y}:=\Lambda_{X,Y}\text{ else,} \label{e-redWAM}\\ 
   &\hat{\Lambda}\in\C[W]^{q^{\delta}\times q^{\delta}}\text{ where }\hat{\Lambda}_{0,0}:=\Lambda_{0,0}-1\text{ and }
     \hat{\Lambda}_{X,Y}:=\Lambda_{X,Y}\text{ else.} \label{e-Lambdahat} 
\end{align}
The matrix~$\tilde{\Lambda}$ is the weight adjacency matrix of the state transition diagram
associated with  $(A,B,C,D)$ where all edges of the form $0\,\edge{u}{v}\;0$ have been removed
whereas~$\hat{\Lambda}$ is the weight adjacency matrix of the state transition diagram associated with $(A,B,C,D)$
and where the trivial edge $0\,\edge{0}{0}\;0$ has been removed.
If all Forney indices of the code are positive we have $\hat{\Lambda}=\tilde{\Lambda}$.
Indeed, in this case the matrix~$B$ has full row rank and therefore the only solution to
the state transition $0=0A+uB$ is $u=0$.
With all these data we have the following means to compute various distance parameters.

\begin{prop}\label{P-WAMdistparam}
\begin{alphalist}
\item The column distances satisfy $d^c_j=\min\{\del\big((\hat{\Lambda}\Lambda^{j})_{0,Y}\big)\mid Y\in\F^{\delta}\}$ for all $j\in\N_0$.
\item The weight enumerator of~$\cC$ can be computed from~$\Lambda$ via
      $\we(\cC)= 2-\Phi^{-1}$, where $\Phi=[(I-L\hat{\Lambda})^{-1}]_{0,0}\in\C(W,L)$.
      As a consequence, the extended row distances and the active burst distances can be retrieved from $\WAM(\cC)$ as well.
\item The active column distances and the active segment distances can be computed from $\WAM(\cC)$ as follows:
      \[
        a^c_j=\min\{\del(\tilde{\Lambda}^{j+1}_{0,Y})\mid Y\in\F^{\delta}\},\
        a^s_j=\min\{\del(\tilde{\Lambda}^{j+1}_{X,Y})\mid X,\,Y\in\F^{\delta}\}
      \]
      for $j\in\N_0$ (and $j\geq1$ and not smaller than the smallest Forney index in the case of the active burst
      distances).
      Moreover, the active burst distances satisfy $ a^b_j=\del(\tilde{\Lambda}^{j+1}_{0,0})$.
\end{alphalist}
As a consequence, codes~$\cC$ and~$\cCbar$ satisfying $\WAM(\cC)=\WAM(\cCbar)$ share all the distance
parameters discussed in (a)~--~(c).
\end{prop}
\begin{proof}
As for the consequence, note that, due to Theorem~\ref{T-Lambdaunique}, none of the formulas given in~(a)~--~(c) depends
on the choice of the encoder matrix~$G$.
Thus all these parameters depend only on the invariant $\WAM(\cC)$.
\\[.5ex]
(a) This follows from the fact that~$d^c_j$ is the minimum weight of the codewords associated  with paths of length $j+1$
starting at zero and where the first edge is not the trivial edge $0\edge{0}{0}0$.
\\
(b) This has been shown in \cite[Thm.~3.1]{McE98a}, see also \cite[Thm.~3.8]{GL05p}.
It is easy to see that~$\Phi$ does not depend on the choice of~$G$.
In fact,~$\Phi$ is the weight enumerator for the molecular codewords, see~\cite{McE98a} or
\cite[Proof of Thm.~3.8]{GL05p}.
The rest of~(b) follows from Proposition~\ref{P-extrowactiveburst}.
\\
(c) This can be seen using similar arguments as in~(a).
As for the active segment distance defined in Definition~\ref{D-distparam}(a) one should have in mind
that $v_{[m,m+j]}$ corresponds to the coefficient sequence over the time interval $[m,m+j]$,
where~$m$ is the largest Forney index.
Since one can reach every state $X\in\F^{\delta}$ from the zero state via suitable inputs $u_0,\ldots,u_{m-1}\in\F^k$
within~$m$ steps (see the CCF) this does not put any restriction on the states in the identity above for $a^s_j$.
\end{proof}

\section{Isometries and Monomial Equivalence}\label{SS-Iso}
\setcounter{equation}{0}
In this section we will introduce and discuss various notions of equivalence for convolutional codes.

\begin{defi}\label{D-isometries}
\begin{alphalist}
\item Let $\cC,\,\cCbar\subseteq\F[z]^n$ be two convolutional codes.
      We call~$\cC$ and~$\cCbar$ {\sl isometric\/} if there there exists a weight-preserving $\F[z]$-isomorphism
      $\varphi:\,\cC\longrightarrow\cCbar$, that is, $\wt(v)=\wt(\varphi(v))$ for all $v\in\cC$.
      If, in addition,~$\varphi$ is degree-preserving, that is, $\deg(v)=\deg(\varphi(v))$ for all $v\in\cC$, we
      call~$\varphi$ a {\sl strong isometry\/} and~$\cC$ and~$\cCbar$ are called {\sl strongly isometric}.
\item Two encoder matrices $G,\,\Gbar\in\F[z]^{k\times n}$ are called {\sl isometric\/} (respectively, {\sl strongly isometric\/})
      if the $\F[z]$-isomorphism $\im G\longrightarrow \im\Gbar,\ uG\longmapsto u\Gbar$ is an isometry
      (respectively, strong isometry).
\end{alphalist}
\end{defi}

Despite the obvious fact that isometric codes have the same distance the notion of isometry is too weak  in order to
classify convolutional codes in a meaningful way.
Indeed, the codes $\im(1,1)$ and $\im(1,z)$ in $\F_2[z]^2$ are isometric in the sense above
(in fact, the two given encoders are isometric), but the first code is a block code whereas the second one is not.
This shows that isometric codes do not even share the most basic algebraic parameters.
Due to this weakness it seems to be more reasonable to investigate strongly isometric codes.

Just like in block code theory monomial equivalence will play a central role for our considerations.
For convolutional codes this notion may be defined in various ways.
Consider the following matrix groups
\[
  \cM_n:=\cM_n(\F):=\Bigg\{M\in GL_n(\F)\,\Bigg|\,
              \begin{array}{l} M=PD\text{ for some permutation matrix}\\
                               P\in GL_n(\F) \text{ and some nonsingular diagonal}\\
                               \text{matrix } D\in GL_n(\F)
              \end{array}\Bigg\}
\]
and
\[
  \cM_{n,z}:=\cM_n(\F(z)):=\left\{M\in\F(z)^{n\times n}\,{\Large\mbox{$\Bigg|$}}\,
              \begin{array}{l} M=PD\text{ for some permutation matrix}\\
              P\in GL_n(\F)\text{ and some diagonal matrix}\\
              D=\text{diag}(\alpha_1z^{\mu_1},\ldots,\alpha_nz^{\mu_n})\text{ where }\\
              \alpha_1,\ldots,\alpha_n\in\F^* \text{ and }\mu_1,\ldots,\mu_n\in\Z\end{array}
              \right\}.
\]
The elements of~$\cM_n(\F)$ are called {\sl monomial matrices\/} while the elements of
$\cM_{n,z}$ are called {\sl $z$-monomial matrices.}

\begin{defi}\label{D-ME}
Let~$\cC$ and~$\cCbar\subseteq\F[z]^n$ be two convolutional codes and $G,\,\Gbar\in\F[z]^{k\times n}$ be any basic matrices.
\begin{alphalist}
\item The codes~$\cC$ and~$\cCbar$ are called {\em monomially equivalent\/} (ME) if there exists a monomial matrix
      $M\in \cM_n(\F)$ such that $\cCbar=\{vM\mid v\in\cC\}$.
\item The codes~$\cC$ and~$\cCbar$ are called {\em $z$-monomially equivalent\/} (\zME) if there exists a $z$-monomial matrix
      $M\in \cMzn$ such that $\cCbar=\{vM\mid v\in\cC\}$.
\item We call the matrices $G$ and~$\Gbar$ {\em monomially equivalent\/} ({\em $z$-monomially equivalent}, respectively)
      if there exists a matrix $M\in\cM_n$ ($\cMzn$, respectively) such that $\Gbar=GM$.
\end{alphalist}
\end{defi}

It is obvious that ME implies being strongly isometric and \zME\ implies being isometric.
Moreover, it is obvious that the only isometries on $\F[z]^n$ are given by monomial matrices and thus are even strong
isometries.
Furthermore, the following properties are trivial.

\begin{rem}\label{R-ME}
Let $\cC=\im G$ and $\cCbar=\im\Gbar\subseteq\F[z]^n$ be two codes with encoders~$G$ and~$\Gbar$, respectively.
\begin{arabiclist}
\item The following are equivalent:
      \begin{romanlist}
      \item $\cC$ and~$\cCbar$ are ME (\zME, respectively).
      \item There exist matrices $U\in GL_k(\F[z])$ and $M\in\cM_n$ ($\cMzn$, respectively) such that $\Gbar=UGM$.
      \item There exists an encoder $G'\in\F[z]^{k\times n}$ of~$\cCbar$ such that~$G$ and~$G'$ are ME (\zME, respectively).
      \end{romanlist}
\item If $G,\,\Gbar$ are ME then $\Lambda(G)=\Lambda(\Gbar)$.
      Thus, if $\cC,\,\cCbar$ are ME then $\WAM(\cC)=\WAM(\cCbar)$.
\item Suppose $G,\,\Gbar$ are \zME\ and both are reduced encoders with the same row degrees.
      Then $G,\,\Gbar$ are strongly isometric.
      Indeed, $G,\,\Gbar$ are isometric due to their $z$-monomial equivalence and
      the map $uG\longmapsto u\Gbar$ is degree-preserving due to\eqnref{e-reduced}.
\end{arabiclist}
\end{rem}

Obviously, monomial equivalence and $z$-monomial equivalence are equivalence relations.
However, whereas monomial equivalence can be described in an obvious way by a group action this is not the case for
$z$-monomial equivalence.
This is due to the fact that a matrix product~$GM$ need not be polynomial (and basic) for a given
basic polynomial matrix~$G$ and a given $M\in\cMzn$.

We are now in a position to quote two important results that will be particularly helpful to us later on.
The first one is a version of the famous MacWilliams Equivalence Theorem for block codes.
It completely describes the isometries between block codes.
For the proof see \cite[Sec.~7.9]{HP03} and \cite[Lem.~5.4]{GL05p}.

\begin{theo}\label{T-MacWE}
Let $G,\,\Gbar\in\F^{k\times n}$ be two constant matrices (not necessarily having rank~$k$) such that
$\wt(uG)=\wt(u\Gbar)$ for all $u\in\F^k$.
Then $G$ and~$\Gbar$ are ME.
As a consequence, isometric block codes are ME.
Stated differently, any isometry between block codes in~$\F^n$ extends to an isometry on $\F^n$.
\end{theo}

Due to the last part, this result also came to be known as the MacWilliams Extension Theorem.
Using this classical result one can establish the following result on a particular class of convolutional codes.

\begin{theo}\label{T-posFI}
Let $G,\,\Gbar\in\F[z]^{k\times n}$ be two reduced encoders such that all Forney indices of~$G$ and~$\Gbar$ are positive.
Then $G,\,\Gbar$ are ME if and only if $\Lambda(G)=\Lambda(\Gbar)$.
\\
As a consequence, if $\cC,\,\cCbar\subseteq\F[z]^n$ are codes such that all Forney indices of~$\cC$ and~$\cCbar$
are positive, then $\cC$ and $\cCbar$ are ME if and only if $\WAM(\cC)=\WAM(\cCbar)$.
\end{theo}

The only-if part of the first statement is obvious and was mentioned already in Remark~\ref{R-ME}(2).
The converse is the non-trivial implication and has been proven in \cite[Thm.~3.7]{GS07}.
That proof also implies the consequence for the codes~$\cC$ and~$\cCbar$.
Notice that the result in Theorem~\ref{T-posFI} applies in particular to all codes of dimension one.
One should also observe that, for general dimension, the if-part is not true for block codes, see \cite[Exa.~1.61]{HP03},
and indeed, due to the positivity of the Forney indices, block codes are excluded from the result.

Let us now turn to investigating basic properties of (strong) isometries.

\begin{prop}\label{P-DelayPres}
For a nonzero polynomial vector $v=\sum_{t=0}^M v_tz^t\in\F[z]^n$, $v_t\in\F^n$, define the {\sl delay\/} of~$v$
as $\del(v):=\min\{t\in\N_0\mid v_t\not=0\}$ and put $\del(0):=\infty$.
Then every isomorphism $\varphi:\cC\longrightarrow\cC'$ between codes $\cC,\,\cC'\subseteq\F[z]^n$
is delay-preserving, that is, $\del(v)=\del\big(\varphi(v)\big)\text{ for all }v\in\cC$.
\end{prop}
\begin{proof}
Assume $\cC=\im G$ for some basic matrix~$G\in\F[z]^{k\times n}$.
Each vector $v\in\cC$ such that $\del(v)=\alpha\in\N_0$ can be written as
$v=z^{\alpha}\hat{v}$ where $\hat{v}\in\F[z]^n$ and $\del(\hat{v})=0$.
Due to\eqnref{e-Gbasic} the matrix $G(0)\in\F^{k\times n}$ has full row rank  and thus
$v=z^{\alpha}\hat{v}=uG$ implies $u=z^{\alpha}\hat{u}$ for some $\hat{u}\in\F[z]^k$.
Hence $\hat{v}=\hat{u}G\in\cC$.
Now we have $\varphi(v)=z^{\alpha}\varphi(\hat{v})$ and thus
$\del(\varphi(v))=\alpha+\del(\varphi(\hat{v}))\geq\alpha=\del(v)$.
Since this is true for all $v\in\cC$ and~$\varphi$ is an isomorphism, symmetry between $\cC$ and $\cC'$
implies $\del(v)=\del(\varphi(v))$ for all $v\in\cC$.
\end{proof}

The result above is not true if we replace the codes by arbitrary $\F[z]$-submodules of $\F[z]^n$.
Indeed, the map $v\longmapsto zv$ is an isomorphism between $\im(1,1)$ and $\im(z,z)$ that is not delay-preserving.

\begin{prop}\label{P-StrIso}
Let $\cC,\,\cCbar\subseteq\F[z]^n$ be two codes and let $\varphi:\cC\longrightarrow\cCbar$ be an $\F[z]$-isomorphism.
Suppose $\cC=\im G$ for some encoder $G\in\F[z]^{k\times n}$.
Then
\begin{alphalist}
\item there exists an encoder matrix~$\Gbar\in\F[z]^{k\times n}$ of~$\cCbar$ such that $\varphi(uG)=u\Gbar$ for all $u\in\F[z]^k$;
\item if~$\varphi$ is degree-preserving and~$G$ is reduced, then the matrix~$\Gbar$ in part~(a) is reduced and therefore~$\cC$
      and~$\cCbar$ share the same Forney indices and the same degree.
\end{alphalist}
As a consequence, strongly isometric codes have strongly isometric reduced encoders.
\end{prop}
\begin{proof}
(a) This is a simple consequence of Linear Algebra.
Indeed, let $f:\;\F[z]^k\longrightarrow\cC$, $u\longmapsto uG$ be the encoder map corresponding to~$G$.
Then $g:=\varphi\circ f:\;\F[z]^k\longrightarrow\cCbar,\ u\longmapsto\varphi(uG)$ is an $\F[z]$-isomorphism.
Let~$\Gbar$ be the standard basis representation of~$g$, that is, $g(u)=u\Gbar$ for all $u\in\F[z]^k$.
Since~$g$ is an isomorphism the matrix~$\Gbar$ is indeed an encoder of~$\cCbar$ and
$\varphi(uG)=g(u)=u\Gbar$ for all $u\in\F[z]^k$.
\\
(b) Denote the rows of~$G$ and~$\Gbar$ by $g_1,\ldots,g_k$ and $\bar{g}_1,\ldots,\bar{g}_k$, respectively.
Let~$\delta$ and $\bar{\delta}$ be the degrees of~$\cC$ and~$\cCbar$, respectively.
Then the degree-preserving property of~$\varphi$ along with the reducedness of~$G$ yields
$\delta=\sum_{i=1}^k \deg(g_i)=\sum_{i=1}^k\deg(\bar{g}_i)\geq\bar{\delta}$, where the last
inequality follows from the definition of the degree, see \cite[p.~495]{Fo75}.
Now symmetry between~$\cC$ and~$\cCbar$ yields $\delta=\bar{\delta}$.
This in turn implies equality in the previous step and thus~$\Gbar$ is reduced by the very definition of
reducedness.
Thus,~$\cCbar$ has the same Forney indices as~$\cC$.
\end{proof}

This last result enables us to derive some first coding theoretic properties for isometries.

\begin{cor}\label{C-atomic}
\begin{alphalist}
\item Strong isometries map atomic codewords into atomic codewords.
\item Strongly isometric codes share the same weight enumerator, the same extended
      row distances, and the same active burst distances.
\item Strongly isometric reduced encoders have the same active row distances.
\end{alphalist}
\end{cor}
\begin{proof}
(a) Let $\varphi:\cC\longrightarrow\cC'$ be a strong isometry.
By Proposition~\ref{P-StrIso}(b) there exist reduced encoders $G$ and~$G'$ such that
$\varphi(uG)=uG'$ for all messages~$u$.
Let $v=uG\in\cC$ be atomic.
Then $v_0\not=0$ and thus the same is true for $\varphi(v)=uG'$ due to Proposition~\ref{P-DelayPres}.
Suppose that $uG'$ is not atomic.
Then $uG'=\tilde{u}G'+\hat{u}G'$ where $\tilde{u},\ \hat{u}\in\F[z]^k\backslash\{0\}$ and
$\deg(\tilde{u}G')<\del(\hat{u}G')$.
Since~$\varphi$ is degree- and delay-preserving this implies $\deg(\tilde{u}G)<\del(\hat{u}G)$.
Now, $\varphi(\tilde{u}G+\hat{u}G)=\tilde{u}G'+\hat{u}G'=\varphi(v)$ along with bijectivity yields
$v=\tilde{u}G+\hat{u}G$, contradicting atomicity of~$v$.
\\
(b) is a direct consequence of~(a) along with the definition of the weight enumerator in Definition~\ref{D-distparam}(b)
and Proposition~\ref{P-extrowactiveburst}.
\\
(c)
According to Proposition~\ref{P-StrIso} strongly isometric codes~$\cC$ and~$\cC'$ have strongly isometric reduced
encoders~$G$ and~$G'$ in, say, $\F[z]^{k\times n}$.
In particular,~$G$ and~$G'$ have the same row degrees and their CCF's are given by quadruples $(A,B,C,D)$ and $(A,B,C',D')$
with the same matrices~$A$ and~$B$.
Thus\eqnref{e-state} shows that each pair of codewords $uG$ and $uG'$ corresponding to the same message $u\in\F[z]^k$
share the same state sequence.
All this implies that the strong isometry $uG\longmapsto uG'$ induces weight- and degree-preserving bijections $\cS_j\longrightarrow\cS'_j$
between the sets $\cS_j,\,j\in\N_0$, defined in Definition~\ref{D-distparam} and their counterparts $\cS'_j$ for the code~$\cC'$.
Now the statement about the active row distances follows directly from Definition~\ref{D-distparam}(a)
because $v_{[0,m+j]}$ is a codeword.
\end{proof}

The rest of this section is devoted to examples showing the limitations of (strong) isometries.
Whereas according to the last result strongly isometric convolutional codes share the same
weight enumerator and many distance parameters they, in general,  do not share the
same column distances, active column and active segment distances.
This is, of course, not surprising because the latter are based on truncated rather than complete codewords.
As a consequence, there is no MacWilliams Equivalence Theorem for strong isometries on the class of
convolutional codes, that is, strongly isometric codes need not be ME.

\begin{exa}\label{E-exa1}
Let $G=(1,\,z,\,z,\,1+z)$ and $G'=(1,\,1,\,1,\,1+z)\in\F_2[z]^{1\times4}$.
Both matrices are basic and reduced.
Obviously, the two given encoders are strongly isometric and hence so are $\cC:=\im G$ and $\cC':=\im G'$ in $\F_2[z]^4$.
As a consequence, Corollary~\ref{C-atomic} yields that the codes share the same weight enumerator and the same extended
row distances, active burst distances, and active row distances (the latter are, in this case,  invariants of the code due to the
uniqueness of the reduced encoder for one-dimensional binary codes).
The remaining families of distance parameters do not coincide.
This can be seen as follows.
The WAM's associated with the given encoders are
\[
  \Lambda:=\Lambda(G)=\begin{pmatrix}1&W^2\\W^3&W^3\end{pmatrix},\
  \Lambda':=\Lambda(G')=\begin{pmatrix}1&W^4\\W&W^3\end{pmatrix}.
\]
Here we order the states in~$\F_2$ as $X_1=0,\,X_2=1$ so that $\Lambda_{i,j}=\Lambda_{X_i,X_j}$ in the sense of
Definition~\ref{D-Lambda} and likewise for~$\Lambda'$.
Since there are no state space isomorphisms to take into account this shows that
$\WAM(\cC)\not=\WAM(\cC')$, see Theorem~\ref{T-Lambdaunique}.
From  Proposition~\ref{P-WAMdistparam}(a) we see that $d^c_0=2$ whereas $d'^c_0=4$ (which is also obvious from the constant coefficient matrices of~$G$ and~$G'$).
All other column distances of both codes are equal to~$5$ (which is the distance of both codes).
It remains to show that~$\cC$ and~$\cC'$ do not share the same active column distances and active segment distances.
Consider the reduced WAM's $\tilde{\Lambda}$ and$\tilde{\Lambda'}$ in the sense of\eqnref{e-redWAM}.
Hence the first entry in~$\Lambda$ and~$\Lambda'$ is replaced by~$0$.
Define the matrices of the delays of the entries of $\tilde{\Lambda}^j$ and $\tilde{\Lambda'}^j$ as
$M_j=\big(\del(\widetilde{\Lambda}^j_{X,Y})\big)_{X,Y\in\F_2}$ and
$M'_j=\big(\del(\widetilde{\Lambda'}^j_{X,Y})\big)_{X,Y\in\F_2}$, respectively.
Hence $M_1=\big(\begin{smallmatrix}\infty&2\\3&3\end{smallmatrix}\big)$ and
$M'_1=\big(\begin{smallmatrix}\infty&4\\1&3\end{smallmatrix}\big)$.
It is easy to see by induction that
\[
  M_j=\begin{pmatrix}5\frac{j}{2}&5\frac{j}{2}\\5\frac{j}{2}+1&5\frac{j}{2}\end{pmatrix}\text{ and }
  M'_j=\begin{pmatrix}5\frac{j}{2}&5\frac{j}{2}+2\\5\frac{j}{2}-1&5\frac{j}{2}\end{pmatrix}\text{ for $j$ even}
\]
and
\[
  M_j=\begin{pmatrix}5\frac{j-1}{2}+3&5\frac{j-1}{2}+2\\5\frac{j-1}{2}+3&5\frac{j-1}{2}+3\end{pmatrix}\text{ and }
  M'_j=\begin{pmatrix}5\frac{j-1}{2}+3&5\frac{j-1}{2}+4\\5\frac{j-1}{2}+1&5\frac{j-1}{2}+3\end{pmatrix}\text{ for $j\geq3$ odd.}
\]
With the aid of Proposition~\ref{P-WAMdistparam}(a) and~(c) this shows that the codes have different active column distances for
even~$j\in\N_0$ and different active segment distances for all $j\in\N_0$.
\end{exa}

The following two examples show that even if a strong isometry preserves all distance parameters it need not be a monomial equivalence.
More precisely, the first example shows that strongly isometric codes that share all distance parameters introduced in
Definition~\ref{D-distparam} need not have the same WAM, and thus need not be ME.
This also tells us that the WAM contains much more detailed information about the code than the list of distance parameters.
The second example even shows that strong isometries that do preserve the WAM need not be monomial equivalences.

\begin{exa}\label{E-StrIsoDistanceParam}
Let $\F=\F_2$ and
\[
  G=\begin{pmatrix}1&1&z&1&1&1&1\\1+z&z&0&0&1&1&1\\1&0&0&0&0&0&0\end{pmatrix},\
  G'=\begin{pmatrix}1&1&1&1&1&z&1&\\1+z&z&0&0&1&z&1\\1&0&0&0&0&0&0\end{pmatrix}
  \in\F[z]^{3\times7}.
\]
Obviously, the matrices are \zME\ and reduced with the same row degrees.
Hence they are strongly isometric by Remark~\ref{R-ME}(3).
As a consequence, the two codes $\cC=\im G$ and $\cC'=\im G'\subseteq\F[z]^7$ are strongly isometric and thus,
due to Corollary~\ref{C-atomic}, share the same weight enumerator, extended row distances and active burst distances,
and~$G$ and~$G'$ have the same active row distances.
Furthermore, for $u=(0,0,1)$ we have $uG=uG'=(1,0,0,0,0,0,0)$ and therefore Definition~\ref {D-distparam}(a) shows that
both codes have all column distances equal to $1=\dist(\cC)=\dist(\cC')$.
As for the remaining distance parameters consider the WAM's associated with~$G$ and~$G'$.
They are given by
\begin{equation}\label{e-Lambda12}
   \Lambda=\!\!{\footnotesize \begin{pmatrix}1+W\!\!&\!\!W^3+W^4\!\!&\!\!W^5+W^6\!\!&\!\!W^2+W^3\\
                          W+W^2\!\!&\!\!W^4+W^5\!\!&\!\!W^4+W^5\!\!&\!\!W+W^2\\
                          W+W^2\!\!&\!\!W^4+W^5\!\!&\!\!W^6+W^7\!\!&\!\!W^3+W^4\\
                          W^2+W^3\!\!&\!\!W^5+W^6\!\!&\!\!W^5+W^6\!\!&\!\!W^2+W^3\end{pmatrix}},\;
   \Lambda'=\!\!{\footnotesize\begin{pmatrix}1+W\!\!&\!\!W^2+W^3\!\!&\!\!W^5+W^6\!\!&\!\!W^3+W^4\\
                          W^2+W^3\!\!&\!\!W^4+W^5\!\!&\!\!W^5+W^6\!\!&\!\!W^3+W^4\\
                          W+W^2\!\!&\!\!W^3+W^4\!\!&\!\!W^6+W^7\!\!&\!\!W^4+W^5\\
                          W+W^2\!\!&\!\!W^3+W^4\!\!&\!\!W^4+W^5\!\!&\!\!W^2+W^3\end{pmatrix}},
\end{equation}
where we order the states in~$\F^2$ as
\begin{equation}\label{e-states}
   X_1=(0,0),\ X_2=(0,1),\ X_3=(1,0),\ X_4=(1,1).
\end{equation}
Thus,~$\Lambda_{i,j}=\Lambda_{X_i,X_j}$ for all $i,\,j=1,\ldots,4$ in the sense of Definition~\ref{D-Lambda}
and likewise for~$\Lambda'$.
Since the entry $W^5+W^6$ appears three times in~$\Lambda$ but only twice in~$\Lambda'$
we conclude that $\WAM(\cC)\not=\WAM(\cC')$, see Theorem~\ref{T-Lambdaunique}, and thus the codes are not ME due to Remark~\ref{R-ME}(2).
Using Proposition~\ref{P-WAMdistparam}(c) and the same methods as in the previous example one can show by lengthy,
but straightforward computations that the two codes share the same active column distances
and active segment distances.
The details are given in the appendix.
\end{exa}

\begin{exa}\label{E-StrIsoMWAM}
Let $\F=\F_2$ and consider the two matrices
\[
   G=\begin{pmatrix}z+1&1&z&0&0&0&0&1\\0&0&1&1&0&0&z&1\\1&1&1&1&1&1&0&0\end{pmatrix},\quad
   \Gbar=\begin{pmatrix}z+1&1&z&0&0&0&0&z\\0&0&1&1&0&0&1&z\\1&1&1&1&1&1&0&0\end{pmatrix}.
\]
It is easy to see that both matrices are basic and reduced and obviously have the same row degrees.
Moreover,~$G$ and~$\Gbar$ are obviously \zME\ and therefore the codes $\cC=\im G$ and $\cCbar=\im\Gbar\subseteq\F[z]^8$
are strongly isometric, see Remark~\ref{R-ME}(3).
It is straightforward to show that the WAM's $\Lambda:=\Lambda(G)$ and $\bar{\Lambda}:=\Lambda(\Gbar)$ are given by
\[
  \Lambda\!=\!\!{\footnotesize\begin{pmatrix}1+W^6\!&\!W^3+W^5\!&\!W^3+W^5\!&\!W^2+W^4\\W+W^7\!&\!W^4+W^6\!&\!W^4+W^6\!&\!W^3+W^5\\
                         W^2+W^4\!&\!W^3+W^5\!&\!W^3+W^5\!&\!W^2+W^4\\W^3+W^5\!&\!W^4+W^6\!&\!W^4+W^6\!&\!W^3+W^5\end{pmatrix}}\!,\
  \bar{\Lambda}\!=\!\!{\footnotesize\begin{pmatrix}1+W^6\!&\!W^3+W^5\!&\!W^2+W^4\!&\!W^3+W^5\\W+W^7\!&\!W^4+W^6\!&\!W^3+W^5\!&\!W^4+W^6\\
                         W^3+W^5\!&\!W^4+W^6\!&\!W^3+W^5\!&\!W^4+W^6\\W^2+W^4\!&\!W^3+W^5\!&\!W^2+W^4\!&\!W^3+W^5\end{pmatrix}},
\]
respectively, where the states in $\F^2$ are ordered as in\eqnref{e-states}.
It is immediate to see that the WAM's are related as $\bar{\Lambda}=P\Lambda P^{-1}$, where~$P$ is the permutation matrix
\[
   P=\begin{pmatrix}1&0&0&0\\0&1&0&0\\0&0&0&1\\0&0&1&0\end{pmatrix}.
\]
That means, reordering the states in\eqnref{e-states} as $\bar{X}_1=(0,0),\,\bar{X}_2=(0,1),\,\bar{X}_3=(1,1),\,\bar{X}_4=(1,0)$
the WAM associated with~$G$ takes the form~$\bar{\Lambda}$.
Obviously, this reordering can be achieved by the state space isomorphism $T=\big(\begin{smallmatrix}1&1\\0&1\end{smallmatrix}\big)$, that is,
$X_iT=\bar{X}_i$ for $i=1,\ldots,4$.
Hence $\Lambda_{XT,YT}=\bar{\Lambda}_{X,Y}$ for all $X,\,Y\in\F^2$ and
Theorem~\ref{T-Lambdaunique} shows that $\WAM(\cC)=\WAM(\cCbar)$.
As a consequence, the codes also share the same weight enumerator and all distance parameters and~$G$ and~$\Gbar$ share the same active row distances, see Proposition~\ref{P-WAMdistparam} and Corollary~\ref{C-atomic}.
But the codes are not ME, as we will show by contradiction.
Suppose~$\cC$ and~$\cCbar$ are ME.
Then there exists a matrix
$U\in GL_3(\F[z])$ such that~$UG$ and~$\Gbar$ are ME.
That means in particular, that~$UG$ has to be reduced with row degrees $1,\,1,\,0$.
Hence, according to\eqnref{e-reduced} the matrix~$U$ has to satisfy
\[
   U=\begin{pmatrix}V&u\\0&1\end{pmatrix}, \text{ where }V\in GL_2(\F)\text{ and }u\in\F[z]^{2\times1}\text{ such that }\deg u\leq 1.
\]
But then the third rows of~$G$ and~$UG$ coincide and the last two columns of~$G$ and~$\Gbar$ show that
$V=T=\big(\begin{smallmatrix}1&1\\0&1\end{smallmatrix}\big)$.
Using all possibilities for $u\in\{a+bz\mid a,\,b\in\F^{2\times1}\}$ results in~$16$ options for the matrix~$U$.
Checking all those options shows that none of them leads to a matrix~$UG$ that is ME to~$\Gbar$.
\end{exa}

In Theorem~\ref{P-StrIsoSameWAM} we will see that strongly isometric reduced encoders~$G$ and~$\Gbar$ with the same WAM,
that is, $\Lambda(G)=\Lambda(\Gbar)$, are ME.
Therefore, the example above had to be constructed using a non-trivial state space isomorphism $T\in GL_2(\F)$.

The last example of this section shows that if~$\cC$ and~$\cCbar$ are isometric, but not {\sl strongly\/} isometric, then
the corresponding consequence given in Proposition~\ref{P-StrIso} is not true anymore.
More precisely, in that case there need not exist reduced encoders~$G$ and~$\Gbar$ of~$\cC$ and~$\cCbar$, respectively, such that
$\wt(uG)=\wt(u\Gbar)$ for all $u\in\F[z]^k$.

\begin{exa}\label{E-nonreduced}\
Let $\F=\F_2$ and $\cC=\im G$ and $\cCbar=\im \Gbar\subseteq\F[z]^3$, where
\[
   G=\begin{pmatrix}z^2+z+1&1&0\\z^2&z+1&z^2\end{pmatrix},\quad
   \Gbar=\begin{pmatrix}z^3+z^2+z&1&0\\z^3&z+1&1\end{pmatrix}.
\]
Both matrices are basic, and~$G$ is reduced whereas~$\Gbar$ is not.
The matrices are \zME\ and the mapping $uG\longmapsto u\Gbar$ is a (non degree-preserving) isometry
between~$\cC$ and~$\cCbar$.
A reduced encoder of~$\cCbar$ is given by
\[
   G'=\begin{pmatrix}z^2+z&z&1\\z^2&z^2+z+1&z+1\end{pmatrix},
\]
both codes have degree~$4$ and Forney indices $2,\,2$.
In order to show that there is no isometry of the form $uG\longmapsto u\tilde{G}$
with reduced encoders~$G$ and~$\tilde{G}$ of~$\cC$ and~$\cCbar$, respectively,
we only need to consider the reduced encoders of~$\cCbar$ where both rows have weight~$4$.
Since each reduced encoder of~$\cCbar$ is of the form $UG'$ where $U\in GL_2(\F)$, see\eqnref{e-reduced}, the only
such encoders are
\[
  \tilde{G}_1:=\begin{pmatrix}1&0\\1&1\end{pmatrix}G'=\begin{pmatrix}z^2+z&z&1\\z&z^2+1&z\end{pmatrix},\
  \tilde{G}_2:=\begin{pmatrix}1&1\\1&0\end{pmatrix}G'=\begin{pmatrix}z&z^2+1&z\\z^2+z&z&1\end{pmatrix}.
\]
Using now $u=(z+1,0)$ we compute
$\wt(uG)=4$, $\wt(u\tilde{G}_1)=6$, and $\wt(u\tilde{G}_2)=8$.
All this shows that there is no isometry of the form $uG\longmapsto u\tilde{G}$ between~$\cC$ and~$\cCbar$ where both
encoders are reduced.
Using Proposition~\ref{P-StrIso} all this shows that the codes are not strongly isometric and
do not possess isometric reduced encoders.
\end{exa}

\section{A MacWilliams Equivalence Theorem for Isometry and \zME}\label{SS-MacWE}
\setcounter{equation}{0}

The main result of this section is a MacWilliams Equivalence Theorem for convolutional codes.
It states that isometric codes are $z$-monomially equivalent.
After proving this result we will turn to some consequences about strongly isometric codes.

\begin{theo}[MacWilliams Equivalence Theorem]\label{T-IsozME}
Let~$G,\,\Gbar\in\F[z]^{k\times n}$ be two matrices, not necessarily of rank~$k$, satisfying
$\wt(uG)=\wt(u\Gbar)$ for all $u\in\F[z]^k$.
Then $G$ and~$\Gbar$ are \zME.
As a consequence, for any two codes $\cC,\,\cCbar\subseteq\F[z]^n$ one has
\[
   \cC,\,\cCbar \text{ are isometric}\Longleftrightarrow
   \cC,\,\cCbar\text{ are \zME.}
\]
\end{theo}
It should be noted that, as opposed to block codes, an isometry between codes $\cC,\,\cCbar\in\F[z]^n$ does not
necessarily extend to an isometry on $\F[z]^n$.
This is due to the fact that $z$-monomial matrices are not necessarily polynomial.
We will come back to this when discussing convolutional codes over larger rings in Section~\ref{SS-Gen}.

\begin{proof}
The consequence follows directly from the first part with the aid of Proposition~\ref{P-StrIso}(a).
As for the first statement, let us start with applying a rescaling by $z$-monomial matrices such that
all nonzero columns of~$G$ and~$\Gbar$ have a nonzero constant term.
Then it remains to show that~$G$ and~$\Gbar$ are ME.
Write
\[
   G=\sum_{t=0}^{\nu} G_tz^t,\quad \Gbar=\sum_{t=0}^{\nu}\Gbar_t z^t,
\]
where $G_t,\,\Gbar_t\in\F^{k\times n}$ and~$\nu$ is the maximal degree appearing in~$G$ and~$\Gbar$.
We will proceed by induction on~$\nu$.
For~$\nu=0$ the statement is the MacWilliams Equivalence Theorem~\ref{T-MacWE} for block codes.
Thus let us assume~$\nu>0$.
After a column permutation we may assume
\begin{equation}\label{e-Gcolumns}
   G=\begin{pmatrix}K_1& K_2&0\end{pmatrix},\ \Gbar=\begin{pmatrix}\bar{K}_1&\bar{K}_2&0\end{pmatrix},
\end{equation}
where the matrix $K_1\in\F[z]^{k\times l}$ contains exactly the columns of~$G$ of degree~$\nu$ while
$K_2\in\F[z]^{k\times m}$ consists of the remaining nonzero columns of~$G$ and the last $n-l-m$ columns are zero.
Notice that $l,\,m$, or $n-l-m$ might be zero.
Likewise $\Gbar$ is partitioned such that $\bar{K}_1\in\F[z]^{k\times\bar{l}}$ contains all columns of~$\Gbar$
of degree~$\nu$ and $\bar{K}_2\in\F[z]^{k\times\bar{m}}$ contains the remaining nonzero columns of~$\Gbar$.
Define the sliding generator matrix
\[
  S_{\nu}(G)=\begin{pmatrix} G_0&G_1&\ldots&G_{\nu}&      &  & \\
                                &G_0&\ldots&G_{\nu-1}&G_{\nu}&  & \\
                                &   &\ddots&\vdots&\vdots&\ddots& \\
                                &   &      & G_0  &G_1   &\ldots&G_{\nu} \end{pmatrix}
  \in\F^{k(\nu+1)\times n(2\nu+1)}.
\]
For any message $u=\sum_{t=0}^\nu u_t z^t\in\F[z]^k$ of degree at most~$\nu$ we have $uG=\sum_{t=0}^{2\nu}v_tz^t$, where
\[
  (v_0,\ldots,v_{2\nu})=(u_0,\ldots,u_{\nu})S_{\nu}(G).
\]
Defining the sliding generator matrix $S_{\nu}(\Gbar)$ analogously, the assumption $\wt(uG)=\wt(u\Gbar)$ yields
\[
   \wt\big(\hat{u}S_{\nu}(G)\big)=\wt\big(\hat{u}S_{\nu}(\Gbar)\big) \text{ for all }\hat{u}\in\F^{k(\nu+1)}.
\]
Hence the MacWilliams Equivalence Theorem~\ref{T-MacWE} for block codes implies that
$S_{\nu}(G)$ and $S_{\nu}(\Gbar)$ are ME.
In other words, the columns of these two matrices are identical up to permutation and rescaling by a nonzero constant factor.
The specification of the columns in\eqnref{e-Gcolumns}tells us that $S_{\nu}(G)$ contains exactly $(2\nu+1)(n-l-m)$ zero columns and
exactly~$l$ columns of the form
\begin{equation}\label{e-Glongcolumns}
   \begin{pmatrix}a_{\nu}&a_{\nu-1}& \ldots&a_1&a_0\end{pmatrix}\T, \text{ where }
   a_0,\ldots,a_{\nu}\in\F^k\text{ and }a_0\not=0\not=a_{\nu}.
\end{equation}
These are the compound coefficient vectors of the columns of~$K_1$.
Likewise, the matrix $S_{\nu}(\Gbar)$ contains exactly $(2\nu+1)(n-\bar{l}-\bar{m})$ zero columns and~$\bar{l}$ columns as in\eqnref{e-Glongcolumns}.
But then the monomial equivalence of $S_{\nu}(G)$ and $S_{\nu}(\Gbar)$ yields $l=\bar{l}$ and $m=\bar{m}$ and
the submatrices consisting of the~$l$ columns in~$S_{\nu}(G)$ and $S_{\nu}(\Gbar)$ of the form\eqnref{e-Glongcolumns} are ME.
This shows that $K_1$ and $\bar{K}_1$ are ME and thus $\wt(uK_1)=\wt(u\bar{K}_1)$ for all $u\in\F[z]^k$.
Since $\wt(uG)=\wt(uK_1)+\wt(uK_2)$ and likewise for~$\Gbar$ the assumption on~$G,\,\Gbar$
implies $\wt(uK_2)=\wt(u\bar{K}_2)$ for all $u\in\F[z]^k$.
Now induction on~$\nu$ establishes that~$K_2$ and~$\bar{K}_2$ are ME as well.
This proves the desired result.
\end{proof}

Notice that in the proof above we actually only used that the map $uG\longmapsto u\Gbar$ is
weight-preserving for all $u\in\F[z]^k$ of degree not exceeding~$\nu$.
An immediate consequence of the previous result is obtained when restricting to codes that are delay-free in every
component in the following sense.

\begin{cor}\label{C-IsoME}
Let $G,\,\Gbar\in\F[z]^{k\times n}$ be encoders that are delay-free in every component, that is, each nonzero column
has a nonzero constant term.
Then~$G$ and~$\Gbar$ are isometric if and only if~$G$ and~$\Gbar$ are ME.
As a consequence, if $\cC,\,\cCbar\subseteq\F[z]^n$ are codes such that one, hence any, encoder matrix is delay-free in
every component, then
\[
  \cC,\ \cCbar \text{ isometric}\Longleftrightarrow \cC,\ \cCbar \text{ ME}.
\]
\end{cor}

\begin{exa}\label{E-notME}
The codes given by the encoders
\[
   G=\begin{pmatrix}1&1&z&z&0&0\\1&1&1&1&1&1\end{pmatrix},\quad
   \Gbar=\begin{pmatrix}z+1&1&z&0&0&0\\1&1&1&1&1&1\end{pmatrix}\in\F_2[z]^{2\times 6}
\]
are not isometric.
Indeed, it is straightforward to show that the codes are not ME,
see \cite[Exa~3.8(b)]{GS08}, and thus we may apply Corollary~\ref{C-IsoME}.
It is worth noting that the codes have the same WAM, see, again, \cite[Exa~3.8(b)]{GS08}.
\end{exa}

Using Corollary~\ref{C-IsoME} we can provide an example showing that the duals of strongly isometric
codes need not even be isometric.
It also shows that the property of being \zME\ is not preserved under taking duals.

\begin{exa}\label{E-StrIsoDual}
Let $\F=\F_2$ and $G=(1,\,z,\,1+z)$ and $\Gbar=(z,\,z,\,1+z)$.
Then the codes $\cC=\im G$ and $\cCbar=\im\Gbar$ are obviously strongly isometric (but not ME).
The dual codes are given by $\cC^{\perp}=\im H$ and $\cCbar^{\perp}=\im\bar{H}$, where
\begin{equation}\label{e-HHbar}
  H=\begin{pmatrix} 1&1&1\\z&1&0\end{pmatrix},\quad
  \bar{H}=\begin{pmatrix}1&1&0\\z&1&z\end{pmatrix}.
\end{equation}
The codes ~$\cC^{\perp}$ and $\cCbar^{\perp}$ are not isometric because if they were
then $\cC^{\perp}$ and
\[
  \hat{\cC}:=\im\begin{pmatrix}1&1&0\\z&1&1\end{pmatrix}
\]
would have to be ME due to Corollary~\ref{C-IsoME}.
But that is impossible because their only nonzero constant codewords have weight~$3$ and~$2$, respectively.
As a consequence,~$\cC^{\perp}$ and~$\cCbar^{\perp}$ are not isometric.
Along the same line of arguments one can show that the duals of the codes in Example~\ref{E-StrIsoMWAM} are not
isometric even though the primary codes are strongly isometric and even have the same WAM.
These examples are due to the fact that \zME\ is not a group action on the set of codes in~$\F[z]^n$.
We will come back to this in Section~\ref{SS-Gen} when discussing convolutional codes over rings where \zME\ is a group action.
\end{exa}

Finally we want to turn to codes that are strongly isometric and share the same WAM.
Recall from Example~\ref{E-StrIsoMWAM} that such codes are, in general, not ME.
This is due to the fact that the WAM of a code is an equivalence class rather than a single matrix.
Indeed, in that example we had $\Lambda(G)\not=\Lambda(\Gbar)$, but $\WAM(\cC)=\WAM(\cCbar)$.
If, however, we restrict to the WAM's of the encoders themselves then we obtain the following stronger result.

\begin{theo}\label{P-StrIsoSameWAM}
Let $G, \Gbar\in\F[z]^{k\times n}$ be two reduced and strongly isometric encoders of degree~$\delta$
and suppose that $\Lambda(G)_{X,Y}=\Lambda(\Gbar)_{X,Y}$ for all $X,Y\in\F^{\delta}$ for the
WAM's associated with~$G$ and~$\Gbar$.
Then~$G$ and~$\Gbar$ are ME.
\end{theo}

The assumption $\Lambda(G)=\Lambda(\Gbar)$ expresses the fact that if we fix the same basis on the common state space
of~$\cC$ and~$\cCbar$ the associated minimal state space realizations share the same WAM.

\begin{proof}
By assumption~$G,\,\Gbar$ have the same row degrees in the same ordering.
If all row degrees are positive then we know from Theorem~\ref{T-posFI} that $G,\,\Gbar$ are ME.
Thus, let us assume that, for some $r<k$, the first~$r$ row degrees of~$G,\,\Gbar$ are positive while the last $k-r$
degrees are zero.
Theorem~\ref{T-IsozME} implies that the matrices~$G$ and~$\Gbar$ are \zME.
After applying a suitable column permutation we may assume without loss of generality that
\[
   G=\begin{pmatrix}G_1&G_2\\ D_1&0\end{pmatrix},\ \Gbar=\begin{pmatrix}\Gbar_1&\Gbar_2\\ \bar{D}_1&0\end{pmatrix}
\]
where $D_1\in\F^{(k-r)\times l},\, \bar{D}_1\in\F^{(k-r)\times\bar{l}}$  have no zero columns and the remaining
matrices are of suitable sizes with entries in $\F[z]$.
Then $G,\,\Gbar$ being \zME\ along with the fact that the last~$k-r$ row degrees of~$G$ and~$\Gbar$
are zero implies that $l=\bar{l}$ and that
\begin{equation}\label{e-ME1}
   \begin{pmatrix}G_1\\ D_1\end{pmatrix}\text{ and }\begin{pmatrix}\Gbar_1\\ \bar{D}_1\end{pmatrix}
   \text{ are ME.}
\end{equation}
Hence we may assume without loss of generality that $G_1=\Gbar_1$ and $D_1=\bar{D}_1$ and it remains to
show that~$G_2$ and~$\Gbar_2$ are ME.
One should bear in mind that $\Lambda(G)$ and $\Lambda(\Gbar)$ do not change after applying a monomial equivalence.
Let us turn to the CCF's of~$G$ and~$\Gbar$ and the associated WAM's.
By the very definition in Proposition~\ref{P-CCF} the CCF's of~$G$ and~$\Gbar$ are given by matrix quadruples
$(A,B,C,D)$ and $(\bar{A},\bar{B},\bar{C},\bar{D})$, where
$A=\bar{A},\, B=\bar{B}=\big(\begin{smallmatrix}\tilde{B}\\ 0\end{smallmatrix}\big)$ with
$\tilde{B}\in\F^{r\times \delta}$ being a matrix of full row rank and where the remaining matrices are of the form
\[
  C=\begin{pmatrix}C_1&C_2\end{pmatrix},\ D=\begin{pmatrix}G_{10}&G_{20}\\D_1&0\end{pmatrix},\
  \bar{C}=\begin{pmatrix}C_1&\bar{C}_2\end{pmatrix},\ \bar{D}=\begin{pmatrix}G_{10}&\Gbar_{20}\\D_1&0\end{pmatrix},
\]
where in each of those four matrices the first block contains~$l$ columns and the remaining matrices
are accordingly.
Notice that $G_{i0}$ and $\Gbar_{i0}$ are the constant coefficient matrices of~$G_i$ and~$\Gbar_i$, respectively, for $i=1,2$.
For brevity, write $\Lambda:=\Lambda(G)$ and $\bar{\Lambda}:=\Lambda(\Gbar)$.
By the very definition of the WAM, see~\ref{D-Lambda}, we have $\Lambda_{X,Y}=\bar{\Lambda}_{X,Y}\not=0$ iff $Y=XA+uB$ for
some $u\in\F^k$.
The latter is the case iff $Y=XA+\tilde{u}\tilde{B}$ for a (unique) $\tilde{u}\in\F^r$.
As a consequence, the nonzero entries of~$\Lambda$ are given by
\begin{align*}
 \Lambda_{X,XA+\tilde{u}\tilde{B}}
     &=\we\{XC+uD\mid u\in\F^k: XA+\tilde{u}\tilde{B}=XA+uB\}\\
     &=\we\{XC+\tilde{u}(G_{10},\,G_{20})+\hat{u}(D_1,0)\mid \hat{u}\in\F^{k-r}\}\\
     &=\we\big((XC_1+\tilde{u}G_{10},\, XC_2+\tilde{u}G_{20})+\im(D_1,\,0)\big)\\
     &=\we\big(XC_1+\tilde{u}G_{10}+\im D_1\big)W^{\text{wt}(XC_2+\tilde{u}G_{20})}
\end{align*}
and likewise $\bar{\Lambda}_{X,XA+\tilde{u}\tilde{B}}=\we(XC_1+\tilde{u}G_{10}+\im D_1)W^{\text{wt}(X\bar{C}_2+\tilde{u}\bar{G}_{20})}$
for all $(X,\tilde{u})\in\F^{\delta+r}$.
Hence $\Lambda=\bar{\Lambda}$ implies
\[
  \wt\bigg((X,\tilde{u})\begin{pmatrix}C_2\\G_{20}\end{pmatrix}\bigg)=
  \wt\bigg((X,\tilde{u})\begin{pmatrix}\bar{C}_2\\\Gbar_{20}\end{pmatrix}\bigg)
  \text{ for all }(X,\tilde{u})\in\F^{\delta+r},
\]
and the MacWilliams Equivalence Theorem~\ref{T-MacWE} for block codes yields that
$\big(\begin{smallmatrix}C_2\\G_{20}\end{smallmatrix}\big)$ and
$\big(\begin{smallmatrix}\bar{C}_2\\\bar{G}_{20}\end{smallmatrix}\big)$ are ME.
Since the columns of these matrices consist exactly of the (appropriately stacked) coefficient vectors of
the columns of~$G_2$ and~$\Gbar_2$, respectively, this means that~$G_2$ and~$\Gbar_2$ are ME.
Together with\eqnref{e-ME1} this shows the desired result.
\end{proof}

\begin{theo}\label{T-StrIsoOnePosFI}
Suppose $\cC,\,\cC'\subseteq\F[z]^n$ are two codes with at most one positive Forney index.
Then the following are equivalent.
\begin{romanlist}
\item $\cC,\,\cC'$ are strongly isometric and $\WAM(\cC)=\WAM(\cC')$.
\item $\cC,\,\cC'$ are ME.
\end{romanlist}
In particular, the equivalence is true for all codes of degree at most~$1$.
\end{theo}

Recall that, in general, neither of the conditions in part~(i) implies the other one of that part.
For instance, the codes in Example~\ref{E-notME} have the same WAM (see \cite[Exa.~3.8(b)]{GS08}, but are
not strongly isometric, while Example~\ref{E-exa1} provides us with strongly isometric codes that do
not share the same WAM.
One should also observe that the result extends Theorem~\ref{T-MacWE} as well as Theorem~\ref{T-posFI}.
In fact, if all Forney indices are zero, then this is just the MacWilliams Equivalence Theorem
for block codes and the second condition in~(i) may be omitted.
If all Forney indices are positive, the result is simply Theorem~\ref{T-posFI} and the first condition in~(i) may be omitted.
On the other hand, Example~\ref{E-StrIsoMWAM} shows that the result above is not true if the codes have more than one
positive Forney index and at least one zero Forney index.

\begin{proof}
The implication~(ii)~$\Rightarrow$~(i) is obvious.
Thus let us assume~(i).
Then Proposition~\ref{P-StrIso} implies that there exist reduced encoder matrices $G,\,G'$ of~$\cC,\,\cC'$,
respectively, such that $uG\longmapsto uG'$ is a strong isometry.
Consequently,~$G$ and~$G'$ have the same Forney indices.
If all Forney indices of~$G$ and~$G'$ are zero, the codes are block codes and the result follows from
Theorem~\ref{T-MacWE}.
Thus, let us assume that the first row of~$G$ and~$G'$ each has degree $\delta>0$, while all other rows
are constant.
Consider the WAM's~$\Lambda(G)$ and~$\Lambda(G')$  of~$G$ and~$G'$, respectively.
The assumption $\WAM(\cC)=\WAM(\cC')$ implies the existence of a matrix $T\in GL_{\delta}(\F)$ such that
\begin{equation}\label{e-LambdaLambda'}
   \Lambda(G)_{X,Y}=\Lambda(G')_{XT,YT} \text{ for all }(X,Y)\in\F^{\delta}\times\F^{\delta},
\end{equation}
see Theorem~\ref{T-Lambdaunique}.
In particular,
\begin{equation}\label{e-LambdaXY=0}
   \Lambda(G)_{X,Y}\not=0\Longleftrightarrow\Lambda(G')_{XT,YT}\not=0.
\end{equation}
In order to exploit this equivalence, we need the CCF's $(A,B,C,D)$ and $(A,B,C', D')$ of~$G$ and~$G'$,
respectively.
Notice that the first two matrices of the CCF's of~$G$
and~$G'$ are both of the form
\begin{equation}\label{e-AB}
   A=\left(\begin{smallmatrix} 0&1& & \\ & &\ddots& \\& & &1\\ & & &0\end{smallmatrix}\right)\in\F^{\delta\times\delta},\
   B=\left(\begin{smallmatrix}1& & & \\ &0& & \\ &&\ddots &\\ & & &0\end{smallmatrix}\right)\in\F^{k\times\delta}.
\end{equation}
Using the definition of $\Lambda(G)_{X,Y}$ in\eqnref{e-LambdaXY} we may rewrite\eqnref{e-LambdaXY=0} as
\[
   (X,Y)\in\im\begin{pmatrix}I&A\\0&B\end{pmatrix}\Longleftrightarrow
   (XT,YT)\in\im\begin{pmatrix}I&A\\0&B\end{pmatrix}\Longleftrightarrow
   (X,Y)\in\im\begin{pmatrix}T^{-1}&AT^{-1}\\0&BT^{-1}\end{pmatrix}
\]
for all $(X,Y)\in\F^{\delta}\times\F^{\delta}$. Hence
\[
   \im\begin{pmatrix}I&A\\0&B\end{pmatrix}=\im\begin{pmatrix}T^{-1}&AT^{-1}\\0&BT^{-1}\end{pmatrix},
\]
which in turn yields
\[
  \begin{pmatrix}S&R\\ V&W\end{pmatrix}\begin{pmatrix}I&A\\0&B\end{pmatrix}
  =\begin{pmatrix}T^{-1}&AT^{-1}\\0&BT^{-1}\end{pmatrix}
  \text{ for some }\begin{pmatrix}S&R\\ V&W\end{pmatrix}\in \F^{(\delta+k)\times(\delta+k)}.
\]
But then $V=0$ and $S=T^{-1}$ as well as
\begin{equation}\label{e-RW}
  RB=AT^{-1}-T^{-1}A\text{ and }WB=BT^{-1}.
\end{equation}
Using the form of the matrices~$A$ and~$B$ given in\eqnref{e-AB} it is straightforward to show that the
existence of matrices $R\in\F^{k\times\delta}$ and
$W\in\F^{k\times k}$ such that\eqnref{e-RW} is satisfied implies that $T=\alpha I_{\delta}$ for some
$\alpha\in\F^*$.
Since $Y=XA+uB$ iff $\alpha Y=\alpha XA+\alpha uB$ and $\wt(XC+uD)=\wt(\alpha XC+\alpha uD)$
we obtain from the very definition of the WAM in Definition~\ref{D-Lambda} that
$\Lambda(G')_{X,Y}=\Lambda(G')_{\alpha X,\alpha Y}=\Lambda(G')_{XT,YT}$ for all $(X,Y)\in\F^{\delta}\times\F^{\delta}$.
Hence\eqnref{e-LambdaLambda'} leads to $\Lambda(G)=\Lambda(G')$ and Theorem~\ref{P-StrIsoSameWAM} yields
the desired result.
\end{proof}

\section{Other Notions of Isometries}\label{SS-Gen}
\setcounter{equation}{0}

The isometries considered in the previous sections were based on the fact that, firstly, codes are defined
as $\F[z]$-modules and, secondly, isometries are expected to respect the module structure.
In this section we will briefly address modifications of both the notion of isometry and the
definition of a convolutional code.
Let us begin with the first option.
If, in Definition~\ref{D-isometries}, we require isomorphisms to be just $\F$-linear rather than $\F[z]$-linear,
then isometry becomes a weaker concept and, as we will see next, does not imply $z$-monomial equivalence anymore.

\begin{exa}\label{E-Flinear}
Let $\F=\F_4=\{0,1,\alpha,\alpha^2=\alpha+1\}$ and consider the matrices
\begin{align*}
  &G=\begin{pmatrix}1+z+z^2&1+\alpha z+\alpha^2 z^2& 1+\alpha^2 z+\alpha z^2\end{pmatrix},\\[.6ex]
  &\Gbar=\begin{pmatrix}1+\alpha z+z^2&1+\alpha^2 z+\alpha^2 z^2& 1+z+\alpha z^2\end{pmatrix}.
\end{align*}
Both matrices are basic and thus define codes $\cC:=\im G$ and $\cCbar:=\im\Gbar\subseteq\F[z]^3$.
Define
\[
    S=\begin{pmatrix}0&1&0\\0&0&1\\1&0&0\end{pmatrix}
    \text{ and }\
    M_t=\left\{\begin{array}{ll} I_3,&\text{if }t\in3\Z,\\[.6ex]
                                 S^2,&\text{if }t\in 3\Z+1,\\[.6ex]
                                 \alpha^2S,&\text{if }t\in 3\Z+2.
        \end{array}\right.
\]
Notice that $M_t\in\cM_3$ for all $t\in\N_0$.
Obviously, the map
\begin{equation}\label{e-varphi1}
  \varphi:\F[z]^3\longrightarrow\F[z]^3,\quad \sum_{t\geq0}v_tz^t\longmapsto \sum_{t\geq0}v_tM_tz^t
\end{equation}
is a degree- and weight-preserving $\F$-automorphism of $\F[z]^3$.
It is straightforward to show that
\begin{equation}\label{e-varphi2}
  \varphi(z^tG)=\left\{\begin{array}{ll} \alpha^2z^t\Gbar,&\text{if } t\in 3\Z+2,\\[.6ex]
                                          z^t\Gbar,&\text{else.}
                \end{array}\right.
\end{equation}
Hence, $\F$-linearity implies $\varphi(uG)=\bar{u}\Gbar$, where for $u=\sum_{t\geq0}u_tz^t\in\F[z]$ one
defines $\bar{u}=\sum_{t\geq0}\bar{u}_tz^t\in\F[z]$ via $\bar{u}_t:=\alpha^2 u_t$ if $t\in 3\Z+2$ and $\bar{u}_t:=u_t$ else.
As a consequence, $\varphi(\cC)=\cCbar$ and these codes are strongly $\F$-isometric in the sense that there is
a degree- and weight-preserving $\F$-isomorphism between them.
It is obvious that the codes are not \zME, and thus not ME.
Hence Theorem~\ref{T-IsozME} and Corollary~\ref{C-IsoME} do not remain valid for $\F$-isometries.
From Theorem~\ref{T-posFI} we conclude that $\WAM(\cC)\not=\WAM(\cCbar)$.
On the other hand, it is easy to verify that the codes have identical weight enumerators.
\end{exa}

It should be noted that, in general the image of a code under an $\F$-isometry need not be an $\F[z]$-submodule and thus
not a code.
Indeed, it is straightforward to show, see also \cite[Thm.~8.2]{Pi88b}, that, for an $\F$-isomorphism~$\varphi$, the
image $\varphi(\cC)$ of a code~$\cC$ is a code only if $z^{-1}\varphi^{-1}z\varphi$ is an $\F$-automorphism of~$\cC$.
Obviously, if this  automorphism is the identity then $\varphi$ is $\F[z]$-linear and we arrive at the situation of the
previous sections.
The example above has been constructed using the theory of cyclic convolutional codes.
Such codes provide us with at least one non-trivial $\F$-automorphism and thus allows for the construction of
Example~\ref{E-Flinear}.
Both codes in that example are $\sigma$-cyclic convolutional codes in the sense of \cite[Def.~2.8]{GS04}, where
the $\F$-automorphism~$\sigma$ on $A=\F[x]/_{\langle x^3-1\rangle}\cong\F^3$ is defined by $\sigma(x)=\alpha x$.
Thus, by definition, these codes are left ideals in the skew-polynomial ring $A[z;\sigma]\cong\F[z]^3$ defined by
$az=z\sigma(a)$ for all $a\in A$.
Translated back into the setting of $\F[z]^3$ this means that $\psi(\cC)=\cC$ where
$\psi(v)=\sum_{t\geq0}\alpha^tv_tSz^t$ for all $v=\sum_{t\geq0}v_tz^t\in\F[z]^3$.
The map~$\varphi$ in\eqnref{e-varphi1} has been constructed such that it satisfies $z^{-1}\varphi^{-1}z\varphi=\psi$.
It has to remain open to future research if one can characterize $\F$-isometries (or degree-preserving $\F$-isometries)
of convolutional codes explicitly.
The work in \cite[Ch.~8]{Pi88b} should provide an excellent starting point for this.

Next we turn to the situation where convolutional codes are defined based on rational functions or Laurent
polynomials (or series) rather than polynomials.
In the sequel let $\cR$ be either of the two rings
\begin{equation}\label{e-BiggerRings}
     \F(z)\subseteq \F(\!(z)\!),
\end{equation}
where $\F(\!(z)\!):=\{\sum_{t=l}^\infty v_tz^t\mid l\in\Z,\ v_t\in\F\text{ for all }t\}$ is the
field of Laurent series over~$\F$.
For any matrix~$G\in\F(z)^{k\times n}$ we define
\begin{equation}\label{e-CCR}
   \cC:=\imR G:=\{uG\mid u\in \cR^k\}\subseteq\cR^n
\end{equation}
to be the $\cR$-{\sl (convolutional) code\/}  generated by~$G$
(remember that it is not common to consider convolutional codes that do not possess a rational encoder).
Since the nonzero polynomials are units in~$\cR$ for either of the rings in\eqnref{e-BiggerRings}
it is clear that each $\cR$-code possesses a basic encoder matrix~$G$, that is, a polynomial encoder satisfying\eqnref{e-Gbasic}.
Moreover, two basic encoders $G,\,G'\in\F[z]^{k\times n}$ satisfy $\imR G =\imR G'$ if and only if
$UG=G'$ for some $U\in GL_k(\F[z])$.
Therefore, we have the same notion of reduced encoders and Forney indices as in the polynomial setting.
The dual of an $\cR$-code~$\cC$ is defined, as usual, as $\cC^{\perp}:=\{w\in\cR^n\mid vw\T=0\text{ for all }v\in\cC\}$.
Finally, we define
\begin{equation}\label{e-cCpol1}
    \cCpol:=\cC\cap\F[z]^n
\end{equation}
to be the space of polynomial codewords.
It is easy to see that if $\cC=\imR(G)$ for some basic encoder~$G\in\F[z]^{k\times n}$ then
\begin{equation}\label{e-cCpol2}
    \cCpol=\imFz G:=\{uG\mid u\in\F[z]^k\},
\end{equation}
thus $\cCpol$ is a convolutional code as defined in the previous sections.
Conversely, one can determine the $\cR$-code $\cC$ from $\cCpol$ via
\begin{equation}\label{e-cCpol3}
   \cC=\{\lambda v\mid \lambda\in\cR,\, v\in\cCpol\}.
\end{equation}
From this it follows immediately that $(\cC^{\perp})_{\text{pol}}=(\cCpol)^{\perp}$ in $\F[z]^n$ for each $\cR$-code~$\cC$.
In other words, the notions of duality for the polynomial and the $\cR$-setting are consistent.
Defining ME and \zME\ just like in Definition~\ref{D-ME} it is obvious that $\cR$-codes
$\cC$ and $\cCbar$ are ME if and only if $\cCpol$ and $\cCbarpol$ are.
For \zME, however, the situation is different.
Indeed, the group $\cMzn$ of $z$-monomial matrices is a subgroup of $GL_n(\cR)$ for any choice of~$\cR$, but not
of $GL_n(\F[z])$.
The impact of this becomes most apparent by the following example.

\begin{exa}\label{E-zMER}
Consider the matrices~$H$ and~$\bar{H}$ defined in\eqnref{e-HHbar}.
In Example~\ref{E-StrIsoDual} it has been shown that the polynomial codes $\cCpol:=\imFz H$ and
$\cCbarpol:=\imFz\bar{H}$ are not \zME.
However, the $\cR$-codes $\cC:=\imR H$ and $\cCbar:=\imR\bar{H}$ are \zME\ because
\[
   U\bar{H}\begin{pmatrix}0&z^{-1}&0\\1&0&0\\0&0&z^{-1}\end{pmatrix}=H
   \text{ for }U=\begin{pmatrix}0&1\\z&0\end{pmatrix}\in GL_2(\cR).
\]
In general, since transformation by a $z$-monomial matrix becomes a group action on the set
of all $\cR$-codes in $\cR^n$, we have just like in the block code case
\begin{equation}\label{e-RdualzME}
   \cC,\,\cCbar\subseteq\cR^n\text{ \zME }\Longrightarrow \cC^{\perp},\,\cCbar^{\perp}\text{ \zME}.
\end{equation}
\end{exa}

Let us now turn to general isometries.
First of all, any vector in $\cR^n$ can be written as $\sum_{t\geq l}^{\infty}v_tz^t$, where $v_t\in\F^n$,
and thus we may define the weight of~$v$ as $\wt(v)=\sum_{t\geq l}\wt(v_t)\in\N_0\cup\{\infty\}$.
Notice that, of course, $\dist(\cC)=\dist(\cCpol)$ for any $\cR$-code~$\cC$.
Moreover, we define the delay of~$v$ as $\del(v):=\min\{t\mid v_t\not=0\}$.
We call a map $\varphi:\cC\longrightarrow\cCbar$ between $\cR$-codes~$\cC$ and~$\cCbar$ an $\cR$-{\sl isometry\/} if
$\varphi$ is a weight-preserving $\cR$-isomorphism.
Obviously, the isometries on $\cR^n$ are given by the $z$-monomial matrices.
We have the following straightforward properties.
\begin{prop}\label{P-IsoR}
Let~$\cR$ be any of the rings in\eqnref{e-BiggerRings} and let $\cC,\,\cCbar\subseteq\cR^n$ be two $k$-dimensional
$\cR$-codes.
Let $G\in\F[z]^{k\times n}$ be basic such that $\cC=\imR G$ and suppose that $\varphi:\cC\longrightarrow\cCbar$
is an $\cR$-isometry. Then
\begin{alphalist}
\item there exists a matrix $\Gbar\in\F[z]^{k\times n}$ such that $\cCbar=\imR\Gbar$ and $\wt(uG)=\wt(u\Gbar)$
      for all $u\in\F[z]^k$;
\item $\cC$ and~$\cCbar$ are \zME;
\item if~$\varphi$ is delay-preserving then $\cCpol$ and $\cCbarpol$ are \zME.
\end{alphalist}
\end{prop}
\begin{proof}
(a) Denote the rows of $G$ by $g_1,\ldots,g_k\in\F[z]^n$.
Then there exists some $\alpha\in\Z$ such that
$\varphi(g_i)=z^{\alpha}\bar{g}_i$ for some $\bar{g}_i\in\F[z]^n$, $i=1,\ldots,k$.
Let~$\Gbar\in\F[z]^{k\times n}$ be the matrix with rows $\bar{g}_1,\ldots,\bar{g}_k$.
It is clear that $\cCbar=\imR\Gbar$ since $z^{\alpha}$ is a unit in~$\cR$.
Moreover, for $u\in\F[z]^k$ we have $\wt(uG)=\wt(\varphi(uG))=\wt(z^{\alpha}u\Gbar)=\wt(u\Gbar)$.
\\
(b) By Theorem~\ref{T-IsozME} the matrices~$G$ and~$\Gbar$ are \zME\ and thus so are the
$\cR$-codes $\cC$ and~$\cCbar$.
\\
(c) From\eqnref{e-cCpol2} we have $\cCpol=\imFz G$.
The delay-preserving property of~$\varphi$ implies that we may choose $\alpha=0$ in part~(a) of this proof.
As a consequence, $\varphi(uG)=u\Gbar$ for all $u\in\F[z]^k$.
But then $\cCbarpol=\imFz\Gbar$.
Indeed, if $\bar{v}\in\cCbarpol$, then $\bar{v}=\varphi(v)$ for some $v\in\cC$,
which must be even in $\cCpol$ since~$\varphi$ is delay-preserving.
Hence $\bar{v}=\varphi(uG)=u\Gbar$ for some $u\in\F[z]^k$ and thus $\bar{v}\in\cCbarpol$.
Now~(b) shows that $\cCpol$ and $\cCbarpol$ are \zME.
\end{proof}

Notice that part~(c) above is not true if~$\varphi$ is not delay-preserving.
Indeed, the data in Example~\ref{E-zMER} provide us with an $\cR$-isometry
$\cCbar\longrightarrow\cC,\ uU\bar{H}\longmapsto uH$ that is not delay-preserving and we know already
that $\cCpol$ and $\cCbarpol$ are not \zME.

Since $z$-monomial matrices are in $GL_n(\cR)$, we have the following situation, resembling that for block codes.

\begin{cor}[MacWilliams Extension Theorem]\label{c-zMER}
Let~$\cR$ be any of the rings in\eqnref{e-BiggerRings}.
Then each $\cR$-isometry $\cC\longrightarrow\cCbar$ between $\cR$-codes $\cC,\,\cCbar\subseteq\cR^n$
extends to an $\cR$-isometry on~$\cR^n$.
\end{cor}

Finally, Proposition~\ref{P-IsoR}(b) along with\eqnref{e-RdualzME} yields the following.
\begin{cor}\label{C-dualR}
If $\cC,\,\cCbar\subseteq\cR^n$ are $\cR$-isometric, then so are~$\cC^{\perp}$ and $\cCbar^{\perp}$.
\end{cor}

It appears natural to call an $\cR$-isometry {\sl strong\/} if it is degree- and delay-preserving.
Recall from Proposition~\ref{P-DelayPres} that in the polynomial setting isometries are always delay-preserving.
Using this concept, Examples~\ref{E-StrIsoDual} and~\ref{E-zMER} show that Corollary~\ref{C-dualR} does not carry over to
strongly $\cR$-isometric codes.
Indeed, the codes $\imR (1,\,z,\,1+z)$ and $\imR(z,\,z,\,1+z)\subseteq\cR^3$ are strongly $\cR$-isometric, but their
duals are not due the results in Example~\ref{E-StrIsoDual} along with Proposition~\ref{P-IsoR}(c).

Despite these last two results we may conclude that, all in all, there is no essential difference concerning module-theoretic
isometries for polynomial versus $\cR$-convolutional codes.

\appendix
\section*{Appendix}
\setcounter{section}{1}
\renewcommand{\theequation}{A.\arabic{equation}}
\setcounter{equation}{0}
We will show that the codes $\cC$ and $\cC'$ in Example~\ref{E-StrIsoDistanceParam} share the same active column and active segment
distances.
We need the reduced WAM's $\widetilde{\Lambda}$ and$\widetilde{\Lambda'}$ in the sense of\eqnref{e-redWAM}.
The weight adjacency matrices associated with the given encoders~$G$ and~$G'$ are given in\eqnref{e-Lambda12}.
Hence in the reduced WAM's $\widetilde{\Lambda}$ and$\widetilde{\Lambda'}$ the first entry in~$\Lambda$ and~$\Lambda'$
is replaced by~$0$.
For $j\in\N$ define the matrices of delays
$M_j=\big(\del(\widetilde{\Lambda}^j_{X,Y})\big)_{X,Y\in\F^2}$ and
$M'_j=\big(\del(\widetilde{\Lambda'}^j_{X,Y})\big)_{X,Y\in\F^2}$.
It is straightforward to show by induction
\[
  M_1=\!\!{\footnotesize\begin{pmatrix}\infty&3&5&2\\1&4&4&1\\1&4&6&3\\2&5&5&2\end{pmatrix}},\,
  M_2=\!\!{\footnotesize\begin{pmatrix}4&7&7&4\\3&4&6&3\\5&4&6&3\\4&5&7&4\end{pmatrix}},\,
  M_3=\!\!{\footnotesize\begin{pmatrix}6&7&9&6\\5&6&8&5\\5&8&8&5\\6&7&9&6\end{pmatrix}},\,
\]
and
\[
  M_j=\!\!{\footnotesize\begin{pmatrix}2j&2j+1&2j+3&2j\\2j-1&2j&2j+2&2j-1\\2j-1&2j&2j+2&2j-1\\2j&2j+1&2j+3&2j\end{pmatrix}}
  \text{ for }j\geq4
\]
as well as
\[
  M'_1=\!\!{\footnotesize\begin{pmatrix}\infty&2&5&3\\2&4&5&3\\1&3&6&4\\1&3&4&2\end{pmatrix}},\
  M'_2=\!\!{\footnotesize\begin{pmatrix}4&6&7&5\\4&4&7&5\\5&3&6&4\\3&3&6&4\end{pmatrix}},\
  M'_3=\!\!{\footnotesize\begin{pmatrix}6&6&9&7\\6&6&9&7\\5&7&8&6\\5&5&8&6\end{pmatrix}}
\]
and
\[
  M'_j=\begin{pmatrix}2j&2j&2j+3&2j+1\\2j&2j&2j+3&2j+1\\2j-1&2j-1&2j+2&2j\\2j-1&2j-1&2j+2&2j\end{pmatrix}
  \text{ for }j\geq4.
\]
Using Proposition~\ref{P-WAMdistparam}(c) we see that the codes have the same active
column distances $a^c_j=a'^c_j=2(j+1)$ and the same active segment distances $a^s_j=a'^s_j=2j+1$
for all $j\in\N_0$.

\bibliographystyle{abbrv}
\bibliography{literatureAK,literatureLZ}
\end{document}